\documentclass[11pt]{article}

\usepackage{float}
\usepackage{multicol} % Multi column
\usepackage{graphicx}
\usepackage{amsmath, amssymb, amsthm}
\usepackage{verbatim} % for verbatim and comment
\usepackage{geometry} % Set page margin
\usepackage{ifpdf}
\usepackage{calc}
\usepackage[labelfont=bf]{caption}

\allowdisplaybreaks[3]      % Allow page breaks in the middle of displayed equations

\newcommand{\thmabove}{8pt}
\newcommand{\thmbelow}{6pt}
\newcommand{\proofbelow}{8pt}

\newtheoremstyle{mythmstyle}
  {\thmabove}   % Space above
  {\thmbelow}   % Space below
  {\itshape}    % Font of theorem body (e.g., \itshape)
  {}            % Indent amount (empty = no indent, \parindent = para indent)
  {\bfseries}   % Thm head font
  {. }          % Punctuation after thm head
  {2.5pt}       % Space after thm head (\newline = linebreak)
  {\thmname{#1}\thmnumber{ #2}\thmnote{ \normalfont (#3)}}   % Thm head spec
\theoremstyle{mythmstyle}
\numberwithin{equation}{section}

\newenvironment{alg}{
    \begin{list}{}{
        \setlength{\itemsep}{2pt}
        \setlength{\parsep}{0pt}
        \setlength{\parskip}{0pt}
        \setlength{\topsep}{1pt}
    }
}
{
    \end{list}
}

\floatstyle{ruled}
\newfloat{algorithm}{t}{loa}
\floatname{algorithm}{Algorithm}

\newtheorem{theorem}{Theorem}[section]
\newtheorem{lemma}[theorem]{Lemma}

\newtheorem{corollary}[theorem]{Corollary}

\newtheorem*{remark}{Remark}
\newtheorem{claim}[theorem]{Claim}
\renewenvironment{proof}{\noindent\textbf{Proof.}\,}{\afterproof}
\newenvironment{proofof}[1]{\noindent\textbf{Proof} \,(of #1).\,}{\afterproof}
\newenvironment{subproof}{\noindent\textit{Proof.}\,}{\aftersubproof}
\newcommand{\afterproof}{\hfill $\blacksquare$ \par \vspace{\proofbelow}}
\newcommand{\aftersubproof}{\hfill $\Box$ \par \vspace{\proofbelow}}

\newcommand{\repeatclaim}[2]{\vspace{6pt}\noindent\textbf{#1. }{\it #2} \vspace{6pt}}

\newcommand{\newterm}[1]{\textit{#1}}
\renewcommand{\th}{\ifmmode{^{\textrm{th}}}\else{\textsuperscript{th}\ }\fi}
\newcommand{\smallfrac}[2]{{\textstyle \frac{#1}{#2}}}

\newcommand{\ceil}[1]{\left\lceil #1 \right\rceil}

\newcommand{\bR}{\mathbb{R}}
\newcommand{\bZ}{\mathbb{Z}}
\newcommand{\cA}{\mathcal{A}}
\newcommand{\cB}{\mathcal{B}}
\newcommand{\cC}{\mathcal{C}}
\newcommand{\cE}{\mathcal{E}}

\newcommand{\cI}{\mathcal{I}}

\newcommand{\tO}{\tilde{O}}
\newcommand{\eps}{\epsilon}
\newcommand{\intersect}{\cap}
\newcommand{\union}{\cup}
\newcommand{\cond}{\operatorname{cond}}

\newcommand{\poly}{\operatorname{poly}}

\newcommand{\abs}[1]{\lvert #1 \rvert}
\newcommand{\Abs}[1]{\left\lvert #1 \right\rvert}
\newcommand{\card}[1]{\abs{#1}}
\newcommand{\set}[1]{\left \{ #1 \right \}}                     % Set notation: { ... }
\newcommand{\setst}[2]{\left\{\; #1 \,:\, #2 \;\right\}}        % Set notation: { ... | ... }

\newcommand{\prob}[1]{\operatorname{Pr}\left[\,#1\,\right]}               % Probability: Pr[ .. ] 
\newcommand{\expect}[1]{\operatorname{E}\left[\,#1\,\right]}              % Expectation: E[ blah ]
\newcommand{\AlgorithmName}[1]{\label{alg:#1}}
\newcommand{\Algorithm}[1]{Algorithm~\ref{alg:#1}}
\newcommand{\AppendixName}[1]{\label{app:#1}}
\newcommand{\Appendix}[1]{Appendix~\ref{app:#1}}
\newcommand{\ClaimName}[1]{\label{clm:#1}}
\newcommand{\Claim}[1]{Claim~\ref{clm:#1}}
\newcommand{\CorollaryName}[1]{\label{cor:#1}}
\newcommand{\Corollary}[1]{Corollary~\ref{cor:#1}}
\newcommand{\EquationName}[1]{\label{eq:#1}}
\newcommand{\Equation}[1]{Eq.~\eqref{eq:#1}}
\newcommand{\FigureName}[1]{\label{fig:#1}}
\newcommand{\Figure}[1]{Figure~\ref{fig:#1}}
\newcommand{\LemmaName}[1]{\label{lem:#1}}
\newcommand{\Lemma}[1]{Lemma~\ref{lem:#1}}

\newcommand{\SectionName}[1]{\label{sec:#1}}
\newcommand{\Section}[1]{Section~\ref{sec:#1}}
\newcommand{\TheoremName}[1]{\label{thm:#1}}
\newcommand{\Theorem}[1]{Theorem~\ref{thm:#1}}

\renewcommand{\thefootnote}{\fnsymbol{footnote}}
\title{Graph Sparsification by Edge-Connectivity \\ and Random Spanning Trees}
\author{Wai Shing Fung\footnotemark[1] \qquad Nicholas J. A. Harvey\footnotemark[1]
}
\date{}
\begin{document}

\maketitle
    \renewcommand{\thefootnote}{\fnsymbol{footnote}}
    \footnotetext[1]{
    University of Waterloo, Department of Combinatorics and Optimization. 
    Supported by an NSERC Discovery Grant.
    Email: \texttt{\{wsfung,harvey\}@uwaterloo.ca}.
    }
    \renewcommand{\thefootnote}{\arabic{footnote}}

\begin{abstract}
We present new approaches to constructing graph sparsifiers --- weighted subgraphs
for which every cut has the same value as the original graph, up to a factor of $(1 \pm \epsilon)$.
Our first approach independently samples each edge $uv$ with probability
inversely proportional to the edge-connectivity between $u$ and $v$.
The fact that this approach produces a sparsifier resolves a question
posed by Bencz\'ur and Karger (2002).
Concurrent work of Hariharan and Panigrahi also resolves this question.
Our second approach constructs a sparsifier by forming the union 
of several uniformly random spanning trees.
Both of our approaches produce sparsifiers with $O(n \log^2(n)/\epsilon^2)$ edges.
Our proofs are based on extensions of Karger's contraction algorithm,
which may be of independent interest.
\end{abstract}

\section{Introduction}

Graph sparsification is an important technique in designing efficient graph algorithms.
Different notions of graph sparsifiers have been considered in the literature.
Roughly speaking, given a graph $G=(V,E)$, a sparsifier $G'$ of $G$
is a sparse subgraph of $G$ that approximates $G$ in some measures, 
e.g., pairwise distance, cut values,
or the quadratic form defined by the graph Laplacian.
$G'$ may be weighted or not.
Throughout this paper, we let $n = |V|$ and $m = |E|$.
Since many graph algorithms have running times that depend on $m$,
if $G$ is dense, then the running time can be improved by replacing $G$ with $G'$,
possibly with some loss in the quality of the solution.

Let us define a \textit{cut sparsifier} to be a weighted subgraph that
approximately preserves the value of every cut to within a multiplicative error of $(1\pm \epsilon)$.
The main motivation for cut sparsifiers was to improve the runtime of
approximation algorithms for finding various sorts of cuts;
indeed, they have been used extensively for this purpose \cite{KargerSkel,BKConf,AK,KRV}.
The first cut sparsifier was Karger's \textit{graph skeleton}~\cite{KargerSkelConf,KargerSkel}.
He showed that sampling each edge independently with probability $p=\Theta(\log n/\epsilon^2 c)$,
where $c$ is the size of the min cut, gives a sparsifier of size $O(pm)$.
Unfortunately, this is of little use when $c$ is small.
The celebrated work of Bencz\'ur and Karger~\cite{BKConf,BK} 
improved on this by using non-uniform sampling,
obtaining a cut sparsifier with only $O(n \log n / \epsilon^2)$ edges.
Their sparsifier is constructed by randomly sampling every edge with probability
inversely proportional to its \emph{edge strength},
and weighting the sampled edges accordingly.

Spielman and Teng \cite{STConf,ST} define \textit{spectral sparsifiers} --- subgraphs
that approximately preserve the quadratic form of the graph Laplacian.
Such sparsifiers are stronger than the previously mentioned sparsifiers that only preserve cuts.
Spielman and Teng's motivation for studying spectral sparsifiers was to use them as a building
block for algorithms that solve linear systems in near-linear time \cite{STConf,SpielmanSurvey}.
They construct spectral sparsifiers with $O( n \log^c n )$ edges, for some large constant $c$.
This was improved to $O(n \log n / \epsilon^2)$ edges by Spielman and Srivastava~\cite{SS},
by independently sampling edges according to their effective resistances.

Spielman and Srivastava conjectured that there exist spectral sparsifiers with
$O(n / \epsilon^2)$ edges.
Towards that conjecture, Goyal, Rademacher and Vempala \cite{GoyalRV09}
showed that sampling just two random spanning trees gives a cut sparsifier
in bounded-degree graphs and random graphs.
Finally, in a remarkable paper, Batson, Spielman and Srivastava \cite{BSS}
construct spectral sparsifiers with only $O(n / \epsilon^2)$ edges.

In this paper, we study several questions provoked by this previous work.
\begin{itemize}
\item Bencz\'ur and Karger ask: Does sampling according to
edge connectivity instead of edge strength give a sparsifier?
\item The subgraph produced by Goyal, Rademacher and Vempala is an \textit{unweighted}
subgraph. If we sample random spanning trees and apply weights to the resulting edges,
does this give a better sparsifier?
\item Are there other approaches to achieving sparsifiers with $o(n \log n)$ edges?
\end{itemize}

In this paper, we give a positive answer to the first two questions.
We also give a negative result on using random spanning trees to answer the third question.
In concurrent, independent work, Hariharan and Panigrahi~\cite{HP} also resolve
the first question.

\subsection{Notation}

Before stating our results, we introduce some notation.
For a multigraph $G=(V,E)$ with edge weights $u : E \rightarrow \bR_+$
and a set $F \subseteq E$ of multiedges,
the notation $u(F)$ denotes $\sum_{e \in F} u_e$.
The notation $|F|$ denotes the total number of copies of all multiedges in $F$.
For any set $S\subseteq V$, we define $\delta(S)$ to be
the set of all copies of edges in $E$ with exactly one end in $S$.
So the notation $u(\delta(S))$ denotes the total weight of the cut $\delta(S)$.

For an edge $st\in E$, the (local) \newterm{edge connectivity} between $s$ and $t$,
denoted $k_{st}$, is defined to be the minimum weight of a cut that separates $s$ and $t$.
The \newterm{effective conductance} of edge $st$, denoted $c_{st}$,
is the amount of current that flows when 
each edge $e$ is viewed as a resistor of value $u_e$
and a unit voltage difference is imposed between $s$ and $t$.
A \newterm{$k$-strong component} of $G$ is a maximal $k$-edge-connected,
vertex-induced subgraph of $G$.
The \newterm{strength} of edge $st$, denoted by $k'_{st}$, is
the maximum value of $k$ such that a $k$-strong component of $G$ contains both $s$ and $t$.

Informally, all three of $k_{st}$, $c_{st}$ and $k'_{st}$
measure the connectivity between $s$ and $t$.
The values of $k'_{st}$ and $c_{st}$ are incomparable:
$k'_{st}$ can be $\Omega(n)$ times larger than $c_{st}$ or vice versa.
However $k_{st} \geq \max \set{ c_{st}, k'_{st} }$ always holds.
For more details, see \Appendix{gapexample}.

\subsection{Our results}

\begin{theorem}
\TheoremName{mainthm}
Let $G=(V,E)$ be a simple, weighted graph with edge weights $u : E \rightarrow \bZ_+$.
Let $\eps \geq 1/n$ and let $\rho := d \log^2 n$ where $d = O(1/\eps^2)$.
For each edge $e$, let $\kappa_e$ be a parameter such that $\kappa_e \leq k_e$.
With high probability,
the graph $G'$ produced by \Algorithm{sample} satisfies
\begin{equation}
    \EquationName{cutspreserved}
    \Abs{u(\delta(S))-w(\delta(S))} ~\le~ \epsilon \cdot u(\delta(S))
        \qquad\forall S\subseteq V.
\end{equation}
Furthermore, with high probability,
\begin{equation}
    \EquationName{sizebound}
    \card{E'} ~\leq~ 2 \rho \sum_{e \in E} \frac{u_e}{k_e} + O(\log n).
\end{equation}
\end{theorem}

\begin{algorithm}
\begin{alg}
\item	\textbf{procedure} Sparsify($G$,\, $u$,\, $\kappa$)
\item	\textbf{input:} A graph $G=(V,E)$
        with edge weights $u : E \rightarrow \bZ_+$,
        and connectivity estimates $\kappa$
\item	\textbf{output:} A graph $G' = (V,E')$ with edge weights $w : E \rightarrow \bR_+$
\item   For $i=1,\ldots,\rho$
    \begin{alg}
    \item   $\rhd$ \textit{We refer to this as the $i\th$ round of sampling}
    \item   For each $e \in E$
        \begin{alg}
        \item   For $j = 1, \ldots, u_e$
            \begin{alg}
            \item   With probability $p_e = 1/\kappa_e$, add edge $e$ to $E'$
                    (if it does not already exist) and increase $w_e$ by $\kappa_e / \rho$
            \end{alg}
        \end{alg}
    \end{alg}
\item   Return $G'$ and $w$
\end{alg}
\caption{A general algorithm for producing a sparsifier of $G$ by sampling edges.}
\AlgorithmName{sample}
\end{algorithm}

This theorem is proven in Sections \ref{sec:sample_with_connectivity} and \ref{sec:splitting_off_alg}.
The condition that $\eps \geq 1/n$ is not really restrictive because
if $\eps < 1/n$ then the theorem is trivial: $G$ is itself a sparsifier with $O(1/\epsilon^2)$ edges.
The condition that the edge weights are integral is not restrictive either.
If the edge weights are any positive real numbers then they can be approximated
arbitrarily well by rational numbers, and these rationals can be scaled up to integers.
This does not affect the conclusion of \Theorem{mainthm}, as it
does not depend on the magnitude of $u$.

The sampled weight of each copy of edge $e$
is a binary random variable that takes value 
$\kappa_e/\rho$ with probability $1/\kappa_e$ and zero otherwise.
When $\kappa_e=k_e$, this random variable has the highest variance,
and therefore the cuts of $G'$ are least concentrated.
So, at least intuitively, the theorem is hardest to prove when $\kappa_e = k_e$,
and the result for smallest $\kappa_e$ values will follow as a corollary. 
This intuition is indeed correct, and we obtain several interesting corollaries of \Theorem{mainthm}
by invoking \Algorithm{sample} with different $\kappa_e$ values.
Proofs are in \Appendix{corollaries}.

\begin{corollary}
\CorollaryName{sample_with_connectivity}
Let $\kappa_e = k_e$. Then \eqref{eq:cutspreserved} holds
and $|E'| = O(n \log^2 n / \epsilon^2)$ with high probability.
\end{corollary}

\begin{corollary}
\CorollaryName{sample_with_conductance}
Let $\kappa_e = c_e$.
Then \eqref{eq:cutspreserved} holds
and $|E'| = O(n \log^2 n / \epsilon^2)$ with high probability.
\end{corollary}

Spielman and Srivastava \cite{SS} prove a related result:
taking $\kappa_e = c_e$, only $O(\log n/\epsilon^2)$ rounds of sampling suffice
for \eqref{eq:cutspreserved} to hold with \emph{constant} probability.
A simple modification of their proof implies \Corollary{sample_with_conductance}.
(See, e.g., Koutis et al.~\cite{KMP}.)
It is unclear whether $O(\log n/\epsilon^2)$ rounds suffice
for \eqref{eq:cutspreserved} to hold with high probability.

\begin{corollary}
\CorollaryName{sample_with_strength}
Let $\kappa_e = k'_e$.
Then \eqref{eq:cutspreserved} holds
and $|E'| = O(n \log^2 n / \epsilon^2)$ with high probability.
\end{corollary}

Bencz\'ur and Karger \cite{BK} prove a stronger result:
taking $\kappa_e = k'_e$, only $O(\log n/\epsilon^2)$ rounds of sampling suffice
for \eqref{eq:cutspreserved} to hold with high probability.

\vspace{6pt}

An important aspect of our analysis is that we rely only on Chernoff bounds.
In contrast, Spielman and Srivastava \cite{SS} use sophisticated concentration bounds for
matrix-valued random variables.
An advantage of Chernoff bounds is that they are very flexible
and have been generalized in many ways.
This flexibility enables us to prove the following result in \Section{spanning_tree}.

\begin{algorithm}
\begin{alg}
\item	\textbf{procedure} SparsifyByTrees($G$,\, $u$)
\item	\textbf{input:} A graph $G=(V,E)$
        with edge weights $u : E \rightarrow \bZ_+$
\item	\textbf{output:} A graph $G' = (V,E')$ with edge weights $w : E \rightarrow \bR_+$
\item   For each $e \in E$, compute the conductance $c_e$
\item   For $i=1,\ldots,\rho$
    \begin{alg}
    \item   $\rhd$ \textit{We refer to this as the $i\th$ round of sampling}
    \item   Let $T$ be a uniformly random spanning tree
    \item   For each $e \in T$
        \begin{alg}
        \item   Add edge $e$ to $E'$
                (if it does not already exist) and increase $w_e$ by $c_e / \rho$
        \end{alg}
    \end{alg}
\item   Return $G'$ and $w$
\end{alg}
\caption{An algorithm for producing a sparsifier of $G$ by sampling random spanning trees.}
\AlgorithmName{sampletrees}
\end{algorithm}

\begin{theorem}
\TheoremName{sample_with_trees}
Let $G=(V,E,u)$ be a graph with edge weights $u : E \rightarrow \bZ_+$,
let $\eps \geq 1/n$, and let $\rho := d \log^2 n$ where $d = O(1/\eps^2)$.
With high probability, the graph $G'$ produced by \Algorithm{sampletrees} satisfies
$$
    |u(\delta(S))-w(\delta(S))| ~\le~ \epsilon \cdot u(\delta(S)) \qquad\forall S\subseteq V.
$$
Clearly $|E'| \leq \rho(n-1) = O( n \log^2 n / \epsilon^2 )$.
\end{theorem}

\paragraph{Counting Small Cuts.}
An important ingredient in the proof of \Theorem{mainthm}
is an extention of Karger's random contraction algorithm
for computing global minimum cuts \cite{KargerContract,KargerStein}.
We describe two variants of this algorithm which introduce the additional ideas 
of splitting off vertices and performing random walks.
The main purpose of these variants is to prove generalizations of Karger's cut-counting
theorem \cite{KargerContract,KargerStein},
which states that the number of cuts of size at most $\alpha$ times the
minimum is less than $n^{2 \alpha}$.
Our generalizations give ``Steiner variants'' of this theorem.
Roughly speaking, we show that, amongst all cuts that separate a certain set of terminals,
the number of size at most $\alpha$ times the minimum is less than $n^{2\alpha}$.

Since our cut-counting result may be of independent interest,
we state it formally now.

\newcommand{\thmcountcutclass}{
    Let $G=(V,E)$ be a graph and let $B \subseteq E$ be arbitrary.
    Suppose that $k_e \geq K$ for every $e \in B$.
    Then, for every real $\alpha \geq 1$,
    $$
    \card{ \setst{ \delta(S) \intersect B }{ S \subseteq V ~\wedge~ \card{\delta(S)} \leq \alpha K } }
    ~<~ n^{2\alpha}.
    $$}
\begin{theorem}
\TheoremName{count_cut_class}
\thmcountcutclass
\end{theorem}

We discuss this theorem in further detail in \Section{splitting_off_alg}.
For now, let us only mention that this theorem reduces to
Karger's cut-counting theorem
by setting $B=E$ and setting $K$ to the global minimum cut value.
In this special case, it states that
the number of $\alpha$-minimum cuts is at most $n^{2\alpha}$.
In concurrent, independent work, Hariharan and Panigrahi \cite{HP}
have also proven \Theorem{count_cut_class}.

\subsection{Algorithms}

In this section we describe several algorithms to efficiently construct sparsifiers.
To make \Algorithm{sample} into a complete algorithm,
the most challenging step is to efficiently compute the $\kappa_e$ values.
This can be done by computing estimates for either $k_e$, $c_e$, or $k'_e$.

\paragraph{Edge Connectivity.}
The simplest approach is to estimate $k_e$.
Several methods for computing such estimates can be found
in the work of Bencz\'ur and Karger \cite{BK}.
In fact, these methods can be significantly simplified
because they were originally designed for estimating $k'_e$,
which is more challenging to estimate than $k_e$.
The following theorems describe how we can combine these methods with \Algorithm{sample}
to efficiently construct sparsifiers.
The resulting algorithms are simple enough for a real-world implementation,
and they have theoretical value too:
they can be used to improve the $O(m\log^3n)$ running time of 
Bencz\'ur and Karger's sampling algorithm to nearly linear time (cf.~\Theorem{alg3}).

These algorithmic results were also described in the earlier work Hariharan and Panigrahi~\cite{HP}.
In fact, their runtime bounds are slightly better, due to a different method of analysis.

\begin{theorem}
\label{thm:alg1}
Given a graph $G$ and edge weights $u : E \rightarrow \bZ_+$
where $\max_e u_e = \poly(n)$, 
a sparsifier $G'$ of size $O(n\log^3n/\epsilon^2)$
that satisfies \eqref{eq:cutspreserved} can be computed in $O(m+n\log^5n/\epsilon^4)$ time.
If $G$ is simple, the time complexity can be reduced to $O(m)$.
\end{theorem}

\begin{theorem}
\label{thm:alg2}
Given a graph $G$ and edge weights $u : E \rightarrow \bZ_+$
where $\max_e u_e = \poly(n)$, 
a sparsifier $G'$ of size $O(n\log^2n/\epsilon^2)$
that satisfies \eqref{eq:cutspreserved} can be computed in $O(m\log^2n+n\log^4 n/\epsilon^4)$ time.
\end{theorem}

Combining Theorems \ref{thm:alg1} and \ref{thm:alg2}
with Bencz\'ur and Karger's algorithm, we can obtain the following result.

\begin{theorem}
\label{thm:alg3}
Given a graph $G$ and edge weights $u : E \rightarrow \bZ_+$
where $\max_e u_e = \poly(n)$,
a sparsifier $G'$ of size $O(n\log n/\epsilon^2)$
that satisfies \eqref{eq:cutspreserved} can be computed in $O(m+n\log^5n/\epsilon^4)$ time.
\end{theorem}

Proofs of these theorems can be found in \Appendix{SketchAlgorithms}.

\paragraph{Effective Conductance.}
As described above, Spielman and Srivastava~\cite{SS} also
construct sparsifiers by sampling according to the effective conductances.
Moreover, they describe an algorithm to approximate the effective conductances
in $m \log^{O(1)} n$ time.
This algorithm can be implemented more efficiently using the recent simplified method of
Koutis, Miller and Peng~\cite{KMP}.
Combining this with \Algorithm{sample}, we can construct a sparsifier
with $O(n \log^2 n / \epsilon^2)$ edges in $\tO(m \log^3 n)$ time.

\paragraph{Random Spanning Trees.}
\Algorithm{sampletrees} can also be implemented efficiently,
although we do not know how to do this in nearly linear time.
The best known algorithms for sampling (approximately) uniform spanning trees
run in $\tO( \min \set{ m \sqrt{n}, n^{2.376} } )$ time~\cite{CMN,KM}.
Combining this with the method described above for approximating the effective conductances
gives an algorithm to compute sparsifiers with $O(n \log^2 n / \epsilon^2)$ edges.
The running time of this algorithm is dominated by the time needed to sample
$O(\log^2 n/\eps^2)$ random spanning trees.
Although this algorithm is not as efficient as those listed above,
the numerous special properties of random spanning trees 
might make it useful in other ways.

\subsection{Limits of sparsification}

In Corollaries \ref{cor:sample_with_connectivity}, \ref{cor:sample_with_conductance}
and \ref{cor:sample_with_strength},
the number of rounds of sampling $\rho$ cannot be decreased to $o(\log n)$.
To see this, consider a path of length $n$ --- with probability tending to $1$
the sampled graph would be disconnected and hence not approximate the original graph.

Sampling random spanning trees overcomes this obstacle since the graph is
connected with probability $1$.
Indeed, Goyal, Rademacher and Vempala \cite{GoyalRV09} show that, for any constant-degree graph,
the unweighted union of just $2$ spanning trees
approximates every cut to within a factor $O(\log n)$.
In \Section{sp_tree_lb} we prove the following negative result for sampling random spanning trees.

\begin{lemma}
\LemmaName{sp_tree_lb}
For any constant $c \geq 1$,
there is a graph such that \Algorithm{sampletrees}
requires $\rho = \Omega(\log n)$ to approximate all cuts within a factor $c$
with constant probability.
\end{lemma}

\section{Sparsifiers by independent sampling}
\SectionName{sample_with_connectivity}

In this section we prove our main result, \Theorem{mainthm}.
Perhaps the most natural approach would be to analyze the probability of poorly
sampling each cut, then union bound over all cuts.
In \Appendix{partitioning} we explain why this simple approach fails,
why Bencz\'ur and Karger \cite{BK} proposed to decompose the graph and separately analyze the
pieces,
and why their approach does not suffice to prove \Theorem{mainthm}.

Our analysis also involves partitioning the graph, but using a different approach.
In fact, a very similar partitioning was used in an earlier proof of Bencz\'ur and Karger 
\cite[\S 3.2]{BKConf} \cite[\S 9.3.2]{BenczurThesis}.
We will partition the graph into subgraphs,
each consisting of edges with roughly equal values of $k_e$.
Formally, we partition $G$ into subgraphs with edge sets $E_1, E_2, \ldots$,
where
$$ E_i ~=~ \setst{ e \in E }{ 2^i \le k_e < 2^{i+1}}. $$
We emphasize that $E_i$ is defined using $k_e$, not $\kappa_e$.

To prove that the weights of all cuts are nearly preserved
(i.e., that \Equation{cutspreserved} holds),
we will use a Chernoff bound to analyze the error contributed to each cut by each subgraph $E_i$.
A union bound allows us to analyze the probability of large deviation
for all cuts simultaneously.
As in previous work \cite{KargerSkel,BK}, 
the key to making this union bound succeed is to show that most cuts are very large,
so their probability of deviation is very small.
This is achieved by our cut-counting theorem, \Theorem{count_cut_class}.

From this point onwards, to simplify our notation, we will no longer think of $G$ as a weighted
graph, but rather think of it as an unweighted multigraph which has
$u_e$ parallel copies of each edge $e$.
The main benefit of this change is that the total weight of a cut
can now be written $|\delta(S)|$ instead of $u(\delta(S))$
since, for multigraphs, the notation $|\delta(S)|$ gives
the total number of copies of all multiedges in $\delta(S)$.
We hope that this choice of notation makes the following proofs easier to read.

The crucial definition for this paper is as follows.
We say that a non-empty set of edges $F \subseteq E_i$ is 
\newterm{induced by a cut} $C$ if $F=C \intersect E_i$.
Any such set $F$ is called a \newterm{cut-induced} subset of $E_i$.
Note that $F$ could be induced by different cuts $C$ and $C'$,
i.e., $C \cap E_i = F = C' \cap E_i$.
For a cut-induced set $F \subseteq E_i$, define
\begin{align}
\EquationName{qdef}
q(F)   ~:=~ \min \setst{ |\delta(S)| }{ S \subseteq V ~\wedge~ \delta(S) \intersect E_i = F }.
\end{align}
So $q(F)$ is the minimum size\footnote{We remind the reader that the notation
$\card{\delta(S)}$ implicitly involves the edge multiplicities.
So, thinking of $G$ as a weighted graph,
$q(F)$ is really the minimum weight of a cut that induces $F$.}
of a cut that induces $F$.
This is an important definition since the amount of error we can allow
when sampling $F$ is naturally constrained by the smallest cut which induces $F$.

We also define a ``normalized'' form of $q(F)$, which is $\alpha(F) := q(F) / 2^i$.
Note that $2^i$ is a lower bound on the size of any cut that intersects $F$,
because every edge $e \in E_i$ has $2^i \leq k_e$.
So we can think of $\alpha(F)$ as a quantity that measures how close $q(F)$
is to the minimum size of any cut that intersects $E_i$.
Clearly $\alpha(F) \geq 1$.

For any set $F$ of edges, the random variable $X_F$ denotes 
the total weight of all sampled copies of the edges in $F$, over all rounds of sampling.
The main challenge in proving \Theorem{mainthm} is to prove concentration for
all $X_F$ where $F$ is a cut-induced subset of some $E_i$.

\subsection{The bad events}
\SectionName{badevents}

Let $F\subseteq E_i$ be a cut-induced set.
We now define three bad events which indicate that the edges in $F$ were not sampled well.
The first two events are:
\begin{align*}
\cA_F:&\qquad |X_F-|F|| > \epsilon |F| \\
\cB_F:&\qquad |X_F-|F|| > \frac{\epsilon q(F)}{\ln n}
\end{align*}
The third event is not needed to analyze unweighted graphs
(i.e., if every edge multiplicity is $u_e = 1$);
it is only needed to deal with arbitrary weights.
The third event is
\begin{align*}
\cC_F:&\qquad
    X_F-|F| ~>~ g^{-1}\Bigg( \frac{ \epsilon^2 q(F) }{ |F| \ln n } \Bigg) \cdot |F|,
\end{align*}
where $g : \bR \rightarrow \bR$ is the function defined by
\begin{equation}
\EquationName{gdef}
g(x) ~=~ (1+x) \ln( 1+x ) - x.
\end{equation}
Note that its derivative is $g'(x) = \ln(1+x)$,
so $g$ is strictly monotonically increasing on $\bR_+$, and hence invertible.
Furthermore, $g^{-1}$ is also strictly monotonically increasing on $\bR_+$,
by the inverse function theorem of calculus.

We will show that,
assuming that these events do not hold (for certain cut-induced sets),
then the weights of all cuts are approximately preserved.
To this end, we bound the probability of the bad events
by the following three claims.
These claims are proven by straightforward applications of Chernoff bounds
in \Appendix{boundproofs}.
Recall that the parameter $d$ in the statement of the claims
satisfies $d = O(1/\eps^2)$, as stated in \Theorem{mainthm}.

\newcommand{\clmlargelambda}{
    Let $F \subseteq E_i$ be a cut-induced set with $q(F) \leq |F| \ln n$.
    Then
    $$
        \prob{\cA_F} ~\le~ 2n^{-d\alpha(F)\epsilon^2/6}.
    $$
}
\begin{claim}
\ClaimName{large_lambda}
\clmlargelambda
\end{claim}

\newcommand{\clmsmalllambda}{
    Let $F \subseteq E_i$ be a cut-induced set with $q(F) > |F| \ln n$.
    Then
    $$
        \prob{\cB_F} ~\le~ 2 n^{-d \alpha(F) \epsilon^2 /8}.
    $$
}
\begin{claim}
\ClaimName{small_lambda}
\clmsmalllambda
\end{claim}

\newcommand{\clmerrorub}{
    For every cut-induced set $F$,
    $$
        \prob{ \cC_F } \leq n^{-d \alpha(F) \eps / 2}.
    $$
}
\begin{claim}
\ClaimName{errorub}
\clmerrorub
\end{claim}

\newcommand{\clmunionbound}{
    By choosing $d = O(1/\eps^2)$ sufficiently large,
    then with high probability, every cut-induced set $F$ satisfies
    \begin{itemize}
    \item if $q(F) \leq |F| \ln n$ then $\cA_F$ does not hold;
    \item if $q(F) > |F| \ln n$ then $\cB_F$ does not hold; and
    \item $\cC_F$ does not hold.
    \end{itemize}
}
\begin{claim}
\ClaimName{union_bound}
\clmunionbound
\end{claim}

The proof of \Claim{union_bound}, given in \Appendix{boundproofs},
is a straightforward modification of an argument of Karger \cite{KargerSkel}.
The only difference with our proof is that we require a result which
bounds the number of small cut-induced sets.
Such a statement is given by \Corollary{count_canonical_cut},
which follows directly from \Theorem{count_cut_class}.

\begin{corollary}
\CorollaryName{count_canonical_cut}
For each $i$ and any real number $\alpha\ge 1$, 
the number of cut-induced sets $F \subseteq E_i$ with $\alpha(F) \leq \alpha$
is less than $n^{2\alpha}$.
\end{corollary}
\begin{proof}
Since every $e \in E_i$ satisfies $k_e \geq 2^i$,
we may apply \Theorem{count_cut_class} with $B=E_i$ and $K=2^i$.
This yields 
$$
\card{ \setst{ \delta(S) \intersect E_i }{ S \subseteq V ~\wedge~ \card{\delta(S)} \leq \alpha 2^i } }
    ~<~ n^{2\alpha}.
$$
Now, by the definition of $\alpha(F)$ and $q(F)$,
for every cut-induced set $F \subseteq E_i$ with $\alpha(F) \leq \alpha$,
there exists a cut $\delta(S)$ such that $\delta(S) \intersect E_i = F$
and $\card{\delta(S)} \leq \alpha 2^i$.
This proves the desired statement.
\end{proof}

\subsection{All cuts are preserved}
\SectionName{allcuts}

In this section we prove that \Equation{cutspreserved} holds.
Recall that the random variable $X_C$ denotes total weight of all sampled edges in $C$.
Our main lemma is

\begin{lemma}
\LemmaName{wtderror}
With high probability, every cut $C$ satisfies $|X_C - |C|| = O(\epsilon |C|)$.
\end{lemma}

For unweighted graphs, the proof of this lemma is quite simple.
For weighted graphs, we require the following three technical claims,
which are proven in \Appendix{calcproof}.

\newcommand{\clmfsmall}{
    Let $F \subseteq E_i$ be a cut-induced set. Then $|F|< n^2 2^i$.
}
\begin{claim}
\ClaimName{Fsmall}
\clmfsmall
\end{claim}

\newcommand{\clmconc}{
    For any integer $d\ge 0$,
    $$
        \sum_{i\le \lg |C|-2 \lg n - d} \card{C \intersect E_i} ~<~ 2^{1-d} \, |C|.
    $$
}
\begin{claim}
\ClaimName{concentrate_in_large_Ci}
\clmconc
\end{claim}

\newcommand{\clmginv}{
    Define $h : \bR \rightarrow \bR$ by
    $$
    h(x) ~=~ \frac{2x}{\ln(1+\sqrt{x})}.
    $$
    Then $g^{-1}(x) \leq h(x)$ for all $x \geq 0$.
}
\begin{claim}
\ClaimName{ginv}
\clmginv
\end{claim}

\vspace{6pt}
\begin{proofof}{\Lemma{wtderror}}
We wish to prove that $|X_C - |C|| \leq O( \eps |C|)$.
We may assume that the conclusions of \Claim{union_bound} hold,
since they hold with high probability.
We use those facts to bound the error contribution from each cut-induced set
$C_i := C \intersect E_i$.

To perform the analysis, we partition the $C_i$'s into three classes,
according to which bad event ($\cA_{C_i}$, $\cB_{C_i}$ or $\cC_{C_i}$)
will be used to analyze the error.
This partitioning depends on the threshold
$$
    t ~:=~ \lg |C| - 4 \lg n - \lg(1/\eps).
$$
The sets of indices are
\begin{align*}
\cI_1 &~=~ \setst{ i }{ \card{C_i} \geq q(C_i)/\ln n } \\
\cI_2 &~=~ \setst{ i }{ 0 < \card{C_i} < q(C_i)/\ln n ~\wedge~ t < i \leq \lg |C| } \\
\cI_3 &~=~ \setst{ i }{ 0 < \card{C_i} < q(C_i)/\ln n ~\wedge~ 0 \leq i \leq t }.
\end{align*}
We remark that
\begin{equation}
    \EquationName{tfirst}
    \lg |C| - t ~\leq~ 5 \lg n
\end{equation}
since we assume that $\epsilon \geq 1/n$.

To analyze the error $X_C - |C|$, we expand it as a sum over cut-induced sets.
\begin{equation}
\EquationName{errorub}
    X_C - |C|
        ~=~ \sum_{i \in \cI_1} (X_{C_i} - |C_i|)
        ~+~ \sum_{i \in \cI_2} (X_{C_i} - |C_i|) 
        ~+~ \sum_{i \in \cI_3} (X_{C_i} - |C_i|) 
\end{equation}

\paragraph{Unweighted graphs.}
For unweighted graphs, the analysis is simple.
Since any cut satisfies $|C| < n^2$, we have $\lg |C| \leq 2 \lg n$ and so $\cI_3 = \emptyset$.
We have assumed that the conclusions of \Claim{union_bound} hold,
so the events $\cA_{C_i}$ and $\cB_{C_i}$ do not occur
(under the stated conditions on $q(C_i)$).
Therefore
\begin{align}
    \nonumber
    |X_C - |C||
        &~\leq~ \sum_{i \in \cI_1} |X_{C_i} - |C_i||
        ~+~ \sum_{i \in \cI_2} |X_{C_i} - |C_i||
        \\\nonumber
        &~\leq~ \sum_{i \in \cI_1} \epsilon |C_i|
        ~+~ \sum_{i \in \cI_2} \frac{ \epsilon q(C_i) }{ \ln n }
        \\\nonumber
        &~\leq~ \epsilon |C|
        ~+~ \frac{\epsilon |C|\card{\cI_2}}{\ln n} 
        \qquad\text{(since $q(C_i) \leq |C|$)} \\
        \EquationName{unwtd}
        &~=~ O(\epsilon |C|)
        \qquad\text{(by \Equation{tfirst}).}
\end{align}
This completes the proof for the unweighted case.

\paragraph{Weighted graphs.}
For weighted graphs the analysis is slightly more complicated because
the number of cut-induced sets $C_i$ that contribute error may be much larger than $\lg n$,
because $\cI_3$ may be non-empty.
To show that the total error is still small, we will need to use the events $\cC_{C_i}$.

Consider \Equation{errorub} again.
The first two sums were analyzed in \Equation{unwtd},
so it suffices to analyze the third sum.
First we prove a lower bound on this sum:
$$
\sum_{i \in \cI_3} (|C_i| - X_{C_i}) ~\leq~ \sum_{i \in \cI_3} |C_i| ~=~ O( \epsilon |C| ),
$$
by \Claim{concentrate_in_large_Ci} and the definition of $t$.

Now we prove an upper bound.
By \Claim{union_bound}, we may assume that the events $\cC_{C_i}$ do not hold.
\begin{align}
    \nonumber
    \sum_{i \in \cI_3} (X_{C_i} - |C_i|) 
        ~\leq~& \sum_{i \in \cI_3} g^{-1}\Bigg( \frac{ \epsilon^2 q(C_i) }{ |C_i| \ln n } \Bigg) \cdot |C_i|
\intertext{We use the fact that $g^{-1}$ is monotonically increasing and that $|C| \geq q(C_i)$.
This yields}
    \nonumber
    ~\leq~& \sum_{i \in \cI_3}
            g^{-1}\Bigg( \frac{ \epsilon |C| }{ |C_i| \ln n } \Bigg) \cdot |C_i|
        \\\nonumber
    ~\leq~& \sum_{i \in \cI_3}
            h\Bigg( \frac{\epsilon |C|}{|C_i| \ln n} \Bigg) |C_i|
            \qquad\text{(by Claim~\ref{clm:ginv})}\\\nonumber
    %~=~& \sum_{i \in \cI_3}
    %        \frac{2 \epsilon |C|}{\ln(n) \ln\Big(1 + \sqrt{\frac{\epsilon |C|}{|C_i| \ln n}}\Big)}
    %        \\\nonumber
    ~=~& O\Bigg(\frac{\epsilon |C|}{\ln n} \Bigg) \sum_{i \in \cI_3}
            \frac{1}{\log \Big(1 + \sqrt{\frac{\epsilon |C|}{|C_i| \ln n}}\Big)}
            \\\nonumber
    ~\leq~& O\Bigg( \frac{ \epsilon |C| }{\log n} \Bigg) \sum_{i \in \cI_3}
            \frac{1}{\lg \Big( \frac{ \epsilon |C| }{ n^2 \lg(n) 2^i } \Big)}
        \qquad\text{(by \Claim{Fsmall})}\\\nonumber
    ~=~& O\Bigg( \frac{\epsilon |C| }{\log n} \Bigg) \sum_{t-i \in \cI_3}
            \frac{1}{\lg \Big( \frac{ \epsilon |C| }{ n^2 \lg(n) 2^{t-i} } \Big)}
            \\\nonumber
    ~\leq~& O\Bigg( \frac{\epsilon |C| }{\log n} \Bigg) \sum_{t-i \in \cI_3}
            \frac{1}{i  + \lg\Big( \frac{ \epsilon |C| }{ n^2 \lg(n) 2^t } \Big)}
    \\\EquationName{harmonic}
    ~\leq~& O\Bigg( \frac{\epsilon |C| }{\log n} \Bigg) \sum_{t-i \in \cI_3}
        \frac{1}{i + \lg n}.
\end{align}
The last inequality holds since
$$
    \frac{\epsilon |C|}{n^2 \lg(n) 2^t}
    ~=~ \frac{\epsilon |C|}{n^2 \lg(n) 2^{\lg |C| - 4 \lg n - \lg(1/\eps)} }
    ~\geq~ \frac{n^2}{\lg n}.
$$
The sum in \Equation{harmonic} is a subseries of a harmonic series with at most $n^2$ terms
(since there are at most ${n\choose 2}$ distinct $\kappa_e$ values)
so the value of this sum is $O(\log n)$.
Thus we have shown that the third sum in \Equation{errorub} is at most $O(\eps |C|)$.
\end{proofof}

\subsection{The size of the sparsifier}

To complete the proof of \Theorem{mainthm},
it remains to show that $G'$ does not have too many edges,
i.e., \Equation{sizebound} holds.
We have
\begin{align*}
\expect{ \card{E'} }
    &~=~ \sum_{\text{edge } e \in E}
        \prob{ \text{at least one copy of $e$ is sampled in one of the rounds} } \\
    &~=~ \sum_{\text{edge } e \in E} 1 - \Big( 1 - \frac{1}{\kappa_e} \Big)^{\rho u_e} \\
    &~\leq~ \rho \sum_{\text{edge } e \in E} \: \frac{u_e}{\kappa_e},
\end{align*}
so the right-hand side of \Equation{sizebound} is at least $2 \expect{ \card{E'} } + O(\log n)$.
By a simple Chernoff bound, it follows that \Equation{sizebound} holds with high probability.

\section{The cut-counting theorem}
\SectionName{splitting_off_alg}

In this section, we prove \Theorem{count_cut_class}, which is our
generalization of Karger's cut-counting theorem~\cite{KargerContract,KargerStein}.
The proof of Karger's theorem is based on his randomized contraction algorithm for
finding a global minimum cut of a graph.
Roughly speaking he shows that, for any small cut-induced set, it has non-negligible probability
of being output by the algorithm, and hence the number of small cut-induced sets must be small.
We will prove our generalized cut-counting theorem by analyzing a variant of the contraction
algorithm which incorporates the additional idea of splitting off vertices.

The formal statement of our theorem is:

\medskip
\noindent\textbf{\Theorem{count_cut_class}.~}{\it \thmcountcutclass}

This theorem becomes easier to understand by restating it using the terminology of 
\Section{sample_with_connectivity} (cf.~\Corollary{count_canonical_cut}).
A cut-induced subset of $B$ is precisely a set of the form
$\delta(S) \intersect B$, so the theorem is counting cut-induced sets
satisfying some condition.
This condition is: for a cut-induced set $F$,
there must exist $S \subseteq V$ with $\card{\delta(S)} \leq \alpha K$
and $F = \delta(S) \intersect B$.
This condition is equivalent to $q(F) \leq \alpha K$,
where $q$ is the function defined in \Equation{qdef}.
So the conclusion of \Theorem{count_cut_class} can be restated as
$$
\card{ \setst{ F }{ F \text{\normalfont\ is a cut-induced subset of } B ~\wedge~ q(F) \leq \alpha K } }
~<~ n^{2\alpha}.
$$

\paragraph{Comparison to Karger's theorem.}
For the sake of comparison, Karger's theorem is:

\begin{theorem}[Karger \protect\cite{KargerContract,KargerStein}]
\TheoremName{karger}
Let $G=(V,E)$ be a connected graph and let $B \subseteq E$ be arbitrary.
Suppose that the value of the global minimum cut is at least $K$.
Then, for every real $\alpha \geq 1$,
$$
\card{ \setst{ C }{ C \text{\normalfont\ is a cut in } G ~\wedge~ \card{C} \leq \alpha K } }
~<~ n^{2\alpha}.
$$
\end{theorem}

Our theorem improves on Karger's theorem in two ways.
First of all, we count cut-induced sets instead of cuts.
This is clearly more general and, as we mentioned before,
it is useful because it avoids overcounting cut-induced sets that are shared by many cuts.
Secondly, we want to bound the number of ``small'' cut-induced sets in $B$.
The bounds given by both theorem are $n^{2\alpha}$,
where $\alpha$ measures how small a cut or a cut-induced set is.
However in our cut-counting lemma, $\alpha$ is 
relative to $K$, the size of a smallest cut that intersects with $B$,
not relative to the size of a global minimum cut as in Karger's theorem.
This is an improvement since the global minimum cuts may not intersect $B$ at all,
so the global minimum cut value could be much smaller than $K$. 

For concreteness, consider the example in \Figure{example1},
which appears in \Appendix{gapexample}.
Suppose we want to bound the number of cuts that intersect with $E_{\lg n}$.
(Here $E_{\lg n}$ consists of the single edge $st$.)
Note that all such cuts have size $n-1$.
However the global minimum cuts all have size $2$, and they do not intersect with $E_{\lg n}$.
From \Theorem{karger} we see that there are at most $n^{n-1}$ cuts of size at most $n-1$
that intersect $E_{\lg n}$.
In contrast, \Theorem{count_cut_class} states that
there are at most $n^2$ cut-induced subsets of $E_{\lg n}$ 
that are induced by cuts of size at most $n-1$.

\paragraph{A weaker theorem based on effective conductance.}
We have also proven a weaker version of \Theorem{count_cut_class} which does not
suffice to prove \Theorem{mainthm} but does suffice to prove
\Corollary{sample_with_conductance}.
This weaker version is stated as \Theorem{count_cut_2};
it is weaker because the hypothesis $c_e \geq K$ is stronger than the hypothesis $k_e \geq K$,
by \Claim{ck}.

\begin{theorem}
\TheoremName{count_cut_2}
Let $G=(V,E)$ be a graph and let $B \subseteq E$ be arbitrary.
Suppose that $c_e \geq K$ for every $e \in B$.
Then, for every real $\alpha \geq 1$,
$$
\card{ \setst{ F }{ F \text{\normalfont\ is a cut-induced subset of } B ~\wedge~ q(F) \leq \alpha K } }
~<~ n^{2\alpha}.
$$
\end{theorem}

Although this theorem is weaker than \Theorem{count_cut_class},
we feel that it is worth including in this paper because its proof is based on analysis of
random walks that may be of independent interest.
The proof is in \Appendix{random_walk_alg}.

\subsection{The generalized contraction algorithm}

\Theorem{count_cut_class} follows immediately from \Theorem{contract_with_split},
which is the analysis of our generalized contraction algorithm
(\Algorithm{contract}).
Henceforth, we will use the following terminology.
The edges in $B$ are called \newterm{black}.
Also, a cut is \newterm{black} if it contains a black edge,
and a vertex is black if it is incident to a black edge.
An edge, vertex or cut is \newterm{white} if it is not black.

\begin{theorem}
\TheoremName{contract_with_split}
For any cut-induced set $F \subseteq B$ with $q(F) \leq \alpha K$,
\Algorithm{contract} outputs $F$ with probability at least $n^{-2\alpha}$.
\end{theorem}

\begin{algorithm}
\begin{alg}
\item	\textbf{procedure} Contract($G$,\, $B$,\, $\alpha$)
\item	\textbf{input:} A graph $G=(V,E)$, a set $B \subseteq E$, and an approximation factor $\alpha$
\item	\textbf{output:} A cut-induced subset of $B$
\item	While there are more than $\lceil 2\alpha \rceil$ vertices remaining
    \begin{alg}
    \item	While there exists a white vertex $v$
        \begin{alg}
        \item   Perform admissible splitting-off at $v$ until $v$ becomes an isolated vertex
        \item   Remove $v$
        \end{alg}
    \item	Pick an edge $e$ uniformly at random
    \item	Contract $e$ and remove any self loops
    \end{alg}
\item	Pick a non-empty proper subset $S$ of $V$ uniformly at random
        and output the black edges with exactly one endpoint in $S$
\end{alg}
\caption{An algorithm for finding a small cut-induced set by splitting off white vertices.}
\AlgorithmName{contract}
\end{algorithm}

\Algorithm{contract} is essentially the same as Karger's contraction algorithm,
except that it maintains the invariant that $G$ has no white vertex by \newterm{splitting-off}.
For a pair of edges $uv$ and $vw$, splitting-off $uv$ and $vw$
is the operation that removes $uv$ and $vw$ then adds a new edge $uw$.
This splitting-off operation is \newterm{admissible} 
if it does not decrease the (local) edge connectivity
between any pair of vertices $s$ and $t$, except of course when one of those vertices is $v$.
Splitting-off has many applications in solving connectivity problems
because of the following theorem.

\begin{theorem}[Mader \protect\cite{Mader}]
\TheoremName{mader}
Let $G=(V,E)$ be a connected graph and $v\in V$ be a vertex.
If $v$ has degree $\neq 3$ and $v$ is not incident to any cut edge,
then there is a pair of edges $uv$ and $vw$ such that
the splitting-off of $uv$ and $vw$ is admissible.
\end{theorem}

Since \Algorithm{contract} needs to perform admissible splitting-off,
we must ensure that the hypotheses of Mader's theorem are satisfied.
This can be accomplished by the simple trick of duplicating every edge,
which ensures that $G$ is Eulerian and its components are $2$-edge-connected.
Note that these conditions are preserved under
all modifications to the graph performed by \Algorithm{contract},
namely contraction, splitting-off and removal of self loops.

\newcounter{itemctr}
\newenvironment{invariants}{
    \begin{list}{(I\textrm{\arabic{itemctr}}): }{
        \usecounter{itemctr}
        \setlength{\itemindent}{0pt}
        \setlength{\labelwidth}{24pt}
        \setlength{\labelsep}{9pt}
        \setlength{\leftmargin}{\parindent+\labelwidth+\labelsep}
        \setlength{\itemsep}{3pt}
        \setlength{\topsep}{6pt}
        \setlength{\listparindent}{0pt}
    }
}
{
    \end{list}
}

To prove \Theorem{contract_with_split},
we fix a cut-induced set $F$ with $q(F) \leq \alpha K$.
We will show that, with good probability,
the algorithm maintains the following invariants.
\begin{invariants}
\item $F$ is a cut-induced set in the remaining graph,
\item $q(F) \leq \alpha K$ (where $q(F)$ now minimizes over cuts in the remaining graph), and
\item every remaining black edge $e$ satisfies $k_e \geq K$.
\end{invariants}

The only modifications to the graph made by \Algorithm{contract}
are splitting-off operations, contraction of edges, and removal of self-loops.
Clearly removing self-loops does not affect the invariants.

\begin{claim}
\ClaimName{splitoff}
The admissible splitting-off operations at $v$ do not affect the invariants.
\end{claim}
\begin{proof}
For (I1), note that splitting-off affects only white edges, and all edges in $F$ are black.
For (I2), note that splitting-off only decreases the size of any cut.
For (I3), the edge connectivity between any two black vertices is unaffected
since the splitting-off is admissible and $v$ is white.
\end{proof}

\begin{claim}
\ClaimName{contract}
Let the number of remaining vertices be $r$.
Assuming that the invariants hold,
they will continue to hold after the contraction operation
with probability at least $1-2\alpha/r$.
\end{claim}
\begin{proof}
For (I3), note that
any black cut which exists after the contraction
also existed before the contraction,
so the edge connectivity between any two black vertices cannot decrease.

Now, with respect the graph before the contraction,
let $C$ be a minimum cardinality cut that induces $F$,
i.e., $\card{C}=q(F)$.
We claim that $\Pr[ e \in C ] \geq 2 \alpha / r$,
where the probability is over the random choice of $e$ to be contracted.
To see this, note that every remaining vertex $u$ is black, 
so the cut $\delta(\set{u})$ is a black cut.
By invariant (I3) we have $\delta(\set{u}) \geq K$,
so the number of remaining edges is at least $K r / 2$.
Since $e$ is picked uniformly at random,
$$
\Pr[ e \in C ]
 ~\leq~ \frac{\card{C}}{K r / 2}
 ~=~ \frac{2 q(F)}{K r}
 ~\leq~ \frac{2\alpha}{r},
$$
by (I2).
Let us assume that $e \not \in C$.
Then $F$ is still induced by $C$ after contracting $e$, so (I1) is preserved. 
Furthermore, (I2) is preserved since $\card{C} \leq \alpha K$.
\end{proof}

The following claim completes the proof of \Theorem{contract_with_split}.
We relegate its proof to \Appendix{contract_alg} as it is the same argument used
to prove Karger's theorem \cite{KargerContract,KargerStein}.
(See also Karger~\cite[App.~A]{KargerSkel}, where a slightly more general result is proven.)

\newcommand{\clmoutput}{
    The probability that \Algorithm{contract} outputs $F$ is at least $n^{-2\alpha}$.
}
\begin{claim}
\ClaimName{output}
\clmoutput
\end{claim}

\subsection{Remarks on cactus representations}

A special case of \Theorem{karger} is that any connected graph $G=(V,E)$ has at most
$n^2$ (non-trivial) minimum cuts.
(In fact, the theorem actually proves a bound of $\binom{n}{2}$, which is tight.)
The same fact is implied by a much earlier result of Dinic, Karzanov and Lomonosov~\cite{DKL},
which states that the minimum cuts have a \newterm{cactus representation}.
Fleiner and Frank \cite{FleinerFrank} give a recent exposition of this result.

Dinitz\footnote{E. Dinic and Y. Dinitz are two different transliterations of the same person's name.}
and Vainshtein~\cite{DVConf,DV} generalized this result as follows.
(See also Fleiner and Jord\'an~\cite{FleinerJordan}.)
Let $U \subseteq V$ be a subset of vertices with $\card{U} \geq 2$.
A cut $C = \delta(S)$ is called a \newterm{$U$-cut} if 
the partition $\set{ U \intersect S, U \setminus S }$ of $U$ that it induces
has both parts non-empty.
A $U$-cut $C$ is called minimal if $\card{C}$ is minimal amongst all $U$-cuts.
Two minimal $U$-cuts are called equivalent if they induce the same partition of $U$.
Dinitz and Vainshtein showed that the equivalence classes of minimal $U$-cuts have
a cactus representation.
In particular, there are at most $n^2$ equivalence classes of minimal $U$-cuts.

We now explain how the latter result also follows from \Theorem{count_cut_class}.
Let $K$ be the minimum cardinality of a $U$-cut.
We add dummy edges of weight $\epsilon/n^2$ between all pairs of $U$-vertices
and let $B$ be the set of dummy edges.
Then every dummy edge $e$ has $k_e \geq K$ and every minimal $U$-cut has weight at most
$(1+\epsilon) K$.
By \Theorem{count_cut_class}, the number of cut-induced sets induced by
cuts of size at most $(1+\epsilon)K$ is at most $n^{2(1+\epsilon)}$.
Any two equivalent minimal $U$-cuts induce the same cut-induced subset of $B$,
so the number of equivalence classes is at most $n^{2(1+\epsilon)}$.
Taking $\epsilon \rightarrow 0$ proves that there are at most
$n^2$ equivalence classes of minimal $U$-cuts.

\section{Sparsifiers by uniform random spanning trees}
\label{sec:spanning_tree}

In this section we describe an alternative approach to constructing a graph sparsifier.
Instead of sampling edges independently at random, as was done in
\Section{sample_with_connectivity},
we will sample edges by picking random spanning trees.
The analysis of this sampling proves \Theorem{sample_with_trees}.
The proof is a small modification of the proof in \Section{sample_with_connectivity},
with some differences to handle the dependence in the sampled edges.
The following two lemmas explain why sampling random spanning trees
is similar to sampling according to effective conductances.

Let us introduce some notation.
For an edge $st\in E$, we denote by $r_{st}$ the effective resistance between $s$ and $t$.
This is the inverse of the effective conductance $c_{st}$.

\begin{lemma}
\LemmaName{negcor}
Let $G=(V,E)$ be an unweighted simple graph, and let $T$ be a spanning tree in $G$
chosen uniformly at random.
Let $e_1,\ldots,e_k \in E$ be distinct edges.
Then
\begin{align}
\label{eq:negcor}
\Pr[e_1, \ldots, e_k \in T] &~\leq~ \Pr[ e_1 \in T ] \cdots \Pr[ e_k \in T ] \\
\label{eq:negcor2}
\Pr[e_1, \ldots, e_k \not\in T] &~\leq~ \Pr[ e_1 \not\in T ] \cdots \Pr[ e_k \not\in T ]
\end{align}
\end{lemma}
\begin{proof}
In the case $k=2$, Equation~\ref{eq:negcor} was known to Brooks et al.~\cite[Equation (2.34)]{BSST}.
See also Lyons and Peres~\cite[Exercise 4.3]{LyonsPeres}.
For general $k$, this is a consequence of Theorem 4.5 in Lyons and Peres~\cite{LyonsPeres},
which is a result of Feder and Mihail \cite{FederMihail}.
See also Goyal, Rademacher and Vempala~\cite[Section 3]{GoyalRV09}.
\end{proof}

One useful consequence of \Lemma{negcor} is that concentration inequalities
can be proven for the number of edges in $T$ that lie in any given subset.
The concentration is due to the following theorem:

\begin{theorem}
\TheoremName{concentration}
Let $a_1, a_2, \ldots, a_k$ be reals in $[0,1]$,
and let $X_1, \ldots X_k$ be $\{0,1\}$-valued random variables.
Suppose that
\begin{align*}
\Pr[ \wedge_{i \in I} X_i = 1 ] &~\leq~ \prod_{i \in I} \Pr[ X_i = 1 ] \qquad \forall I\subseteq[k] \\
\Pr[ \wedge_{i \in I} X_i = 0 ] &~\leq~ \prod_{i \in I} \Pr[ X_i = 0 ] \qquad \forall I\subseteq[k].
\end{align*}
Suppose $\mu_1 \leq E[ \sum_i a_i X_i ] \leq \mu_2$.
Then
\begin{align*}
\Pr[ \sum_i a_i X_i \leq (1-\delta) \mu_1 ]
    &~\leq~ e^{-\mu_1 \delta^2 / 2} \\
\Pr[ \sum_i a_i X_i \geq (1+\delta) \mu_2 ]
    &~\leq~ \Big( \frac{e^\delta}{(1+\delta)^{1+\delta}} \Big)^{\mu_2}
\end{align*}
\end{theorem}
\begin{proof}
See Gandhi et al.~\cite[Theorem 3.1]{Gandhi}.
\end{proof}

We will also use the following corollary.

\begin{corollary}
\CorollaryName{concentration}
Assume the same hypotheses as \Theorem{concentration}.
Let $\mu = E[ \sum_i a_i X_i ]$.
Then
\begin{align*}
\Pr[ | \sum_i a_i X_i - \mu | \geq \delta \mu ]
    ~\leq~ 2 e^{ - \frac{ \delta^2 \mu }{ 2 (1 + \delta/3) }}.
\end{align*}
\end{corollary}
\begin{proof}
Follows from \Theorem{concentration} and basic calculus.
See also McDiarmid~\cite[Theorem 2.3]{McDiarmid}.
\end{proof}

Now consider the approach of \Algorithm{sampletrees} for constructing a sparsifier.
In each round of sampling, instead of picking edges independently,
we pick a uniformly random spanning tree.
Every edge $e$ in the tree is assigned weight $c_e$.
This sampling is repeated for $\rho$ rounds,
and the sparsifier is the sum of these weighted trees.

By \Lemma{kirchoff}, the probability of sampling any particular edge is the same as
when sampling by effective conductances, as was done in \Corollary{sample_with_conductance}.
Furthermore, the same analysis as \Section{sample_with_connectivity} shows that this
sampling method also produces a sparsifier --- the only change to the analysis is that
all uses of Chernoff bounds 
(namely, in Claims \ref{clm:large_lambda}, \ref{clm:small_lambda}
and \ref{clm:errorub})
can be replaced with the concentration bounds in \Theorem{concentration}
and \Corollary{concentration}.
This completes the proof of \Theorem{sample_with_trees}.

\subsection{Lower bound on number of trees}
\SectionName{sp_tree_lb}

In this section, we consider the tradeoff between the number of trees (i.e., the value $\rho$)
and the quality of sparsification in Theorem~\ref{thm:sample_with_trees}.
We prove a lower bound on the number of trees necessary to produce a sparsifier
with a given approximation factor.

\vspace{6pt}
\begin{proofof}{\Lemma{sp_tree_lb}}
Let $G$ be a graph defined as follows.
Its vertices are $\set{u_1,\ldots,u_{n}} \union \set{v_1, \ldots, v_{n+1}}$.
For every $i=1,\ldots,n$, add $k$ parallel edges
$v_i v_{i+1}^{(1)}, \ldots, v_i v_{i+1}^{(k)}$,
and a single length-two path $v_i$-$u_i$-$v_{i+1}$.
The edges $v_i v_{i+1}^{(j)}$ are called \textit{heavy},
and the edges $v_i u_i$ and $u_i v_{i+1}$ are called \textit{light}.
Note that the heavy edges each have effective conductance exactly $(2k+1)/2$.
The light edges each have effective conductance exactly $(2k+1)/(k+1) < 2$.

A uniform random spanning tree in this graph can be constructed by repeating
the following experiment independently for each $i=1,\ldots,n$.
With probability $2k/(2k+1)$, add a uniformly selected heavy edge $v_i v_{i+1}^{(j)}$
to the tree, and add a uniformly selected light edge $v_i u_i$ or $u_i v_{i+1}$ to the tree.
In this case we say that the tree is ``heavy in position $i$''.
Otherwise, with probability $1/(2k+1)$, add both light edges $v_i u_i$ and $u_i v_{i+1}$ to the tree
but no heavy edges.
In this case we say that the tree is ``light in position $i$''.

Our sampling procedure produces a sparsifier that is the union of $\rho$ trees,
where every edge $e$ in the sparsifier is assigned weight $c_e / \rho$.
Suppose there is an $i$ such that all sampled trees are light in position $i$.
Then the cut defined by vertices $\set{v_1,u_1,v_2,u_2,\ldots,v_i}$
has weight exactly $(2k+1)/(k+1) < 2$ in the sparsifier,
whereas the true value of the cut is $k+1$.

The probability that at least one tree is heavy in position $i$ is $1-(2k+1)^{-\rho}$.
The probability that there exists an $i$ such that every tree is light in position $i$ is
$$
p ~:=~ 1-(1-(2k+1)^{-\rho})^n
$$
Suppose $\rho = \ln n/\ln(2k+1)$.
Then $\lim_{n \rightarrow \infty} p = 1-1/e$.
So with constant probability, there is an $i$ such that every tree is light in position $i$,
and so the sparsifier does not approximate the original graph better than a factor
$\frac{k+1}{2}$.
\end{proofof}

\bibliography{Sparsifiers}
\bibliographystyle{plain}
\appendix

\section{Discussion of $k_{st}$, $c_{st}$ and $k'_{st}$}
\AppendixName{gapexample}

As mentioned in the introduction,
the three quantities of an edge $st$ that we consider
(edge connectivity, effective conductance and edge strength)
all roughly measure the connectivity between $s$ and $t$.
However their values can differ significantly.
In this section, we illustrate this with some examples.

Consider a graph which consists of exactly one edge $st$.
To increase $k_{st}$ by $k$,
we can simply add $k$ edge disjoint paths between $s$ and $t$.
In the following examples, we can see that no matter how large $k$ is,
it is possible that $k'_{st}$ or $c_{st}$ increases only by one
while the other increases by $\Omega(k)$.

\begin{figure}
	\centering
    \includegraphics[width=2.5in,clip]{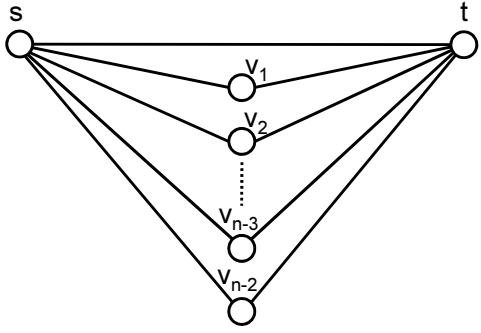}
	\caption{Example showing that conductance can be $\Omega(n)$ times larger than strength}
	\FigureName{example1}
\end{figure}

\begin{itemize}
\item
In \Figure{example1},
$s$ and $t$ are connected by an edge $st$
and $n-2$ paths of length $2$.
Clearly $k_{st}=n-1$, 
$c_{st}=\frac{1}{2}(n-2)+1=n/2$.
But $k'_{st}=2$ as every induced subgraph with at least two vertices is at most $2$ edge connected.

\begin{figure}
	\centering
    \includegraphics[width=2.5in,clip]{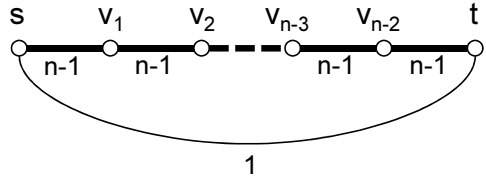}
	\caption{Example showing that strength can be $\Omega(n)$ times larger than conductance}
	\label{example2}
\end{figure}

\item
In Figure \ref{example2},
$s$ and $t$ are connected by an edge $st$
and a path of length $n-1$
which consists of edges of weight $n-1$.
The graph is $n$-edge-connected so
$k_{st}=k'_{st}=n$ but $c_{st}=\frac{n-1}{n-1}+1=2$.
\end{itemize}

Although $c_{st}$ and $k'_{st}$ are incomparable, they are upper bounded by $k_{st}$.

\begin{claim}
\ClaimName{ck}
For any edge $st\in E$, $k_{st} \geq \max \set{ c_{st}, k'_{st} }$.
\end{claim}
\begin{proof}
It is immediate from the definition of edge strength that $k'_{st} \leq k_{st}$,
so we focus on the effective conductance.
Since the connectivity between $s$ and $t$ is $k_{st}$,
there is a cut of size $k_{st}$ separating $s$ and $t$.
Contracting both sides of the cut, 
we get two new vertices $s'$ and $t'$.
By Rayleigh monotonicity~\cite{DoyleSnell}, $c_{s't'}$ is at least $c_{st}$.
Clearly $c_{s't'} = k_{st}$, so the proof is complete.
\end{proof}

\section{Corollaries of \Theorem{mainthm}}
\AppendixName{corollaries}

First we show that our corollaries satisfy the hypotheses of \Theorem{mainthm}.
By \Claim{ck}, Corollaries \ref{cor:sample_with_connectivity}, \ref{cor:sample_with_conductance}
and \ref{cor:sample_with_strength} all have $\kappa_e \leq k_e$,
so \Theorem{mainthm} is applicable.

It remains to analyze $|E'|$, the number of sampled edges.
For Corollaries \ref{cor:sample_with_connectivity} and \ref{cor:sample_with_strength}
we use a property of edge strength proved by Bencz\'ur and Karger \cite[Lemma 2.7]{BK}. 

\begin{lemma}
\LemmaName{sample_size}
In a multigraph with edge strengths $k'_e$, we have
$$
\sum_{e\in E} 1/k'_e ~\le~ n-1.
$$
Here the sum is over all copies of the multiedges.
\end{lemma}

Thus, for Corollaries \ref{cor:sample_with_connectivity} and \ref{cor:sample_with_strength}, we have
$$
\expect{\card{E'}}
    ~=~ \rho \sum_{e\in E} p_e
    ~=~ \rho \sum_{e\in E} 1/\kappa_e
    ~\le~ \rho \sum_{e\in E} 1/k'_e
    ~\le~ \rho (n-1)
    ~=~ O(n \log^2 n/\eps^2).
$$

Finally, we must bound the size of $E'$ in \Corollary{sample_with_conductance}.
We require the following lemma.

\begin{lemma}
\LemmaName{kirchoff}
Let $G=(V,E)$ be a multigraph, and let $T$ be a spanning tree in $G$
chosen uniformly at random.
Then, for any copy of an edge $e \in E$, $\Pr[e \in T] = 1/c_e$.
\end{lemma}
\begin{proof}
See Kirchhoff~\cite{Kirchhoff},
Brooks et al.~\cite[pp.~318]{BSST},
Lov\'asz~\cite[Theorem 4.1(i) and Corollary 4.2]{LovaszSurvey} and
Lyons and Peres~\cite[Corollary 4.4]{LyonsPeres}.
\end{proof}

This immediately implies that
$$
    \sum_{e\in E} 1/c_e ~=~ \sum_{e\in E} \prob{ e \in T } ~=~ \expect{ \card{E(T)} } = n-1.
$$
This fact is known as Foster's theorem,
and it is independently due to Foster~\cite{Foster} and Tetali~\cite{Tetali}.
Thus,
$$
    \expect{ \card{E'} } ~=~ \rho \sum_{e\in E} p_e ~=~ \rho \sum_{e\in E} 1/c_e ~=~ \rho (n-1)
        ~=~ O( n \log^2 n / \eps^2).
$$

\section{Motivation for Partitioning Edges}
\AppendixName{partitioning}

The natural first approach to proving \Theorem{mainthm}
would be to bound the probability of large deviation for each cut
and then union bound over all cuts.
This approach is not feasible, as can be illustrated using the example in \Figure{example3}.

\begin{figure}
	\centering
    \includegraphics[width=2.5in,clip]{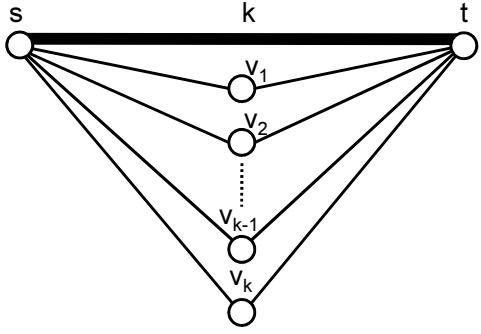}
	\caption{The cut-induced set consisting of $st$ is overcounted $2^n$ times if we simply union
    bound over all cuts.}
	\FigureName{example3}
\end{figure}

In this graph, $s$ and $t$ are connected by $k$ parallel edges
and $k$ paths of length $2$.
Recall that in our sampling scheme 
each copy of $st$ is sampled with probability $1/2k$
for $\rho=O(\log^2n)$ rounds and is assigned a weight of $2k/\rho$ if sampled.
Each edge other than $st$ is sampled with probability $1/2$
for $\rho$ rounds and is assigned a weight of $2/\rho$ if sampled.
Consider a set $U\subseteq V-t$ that contains $s$.
Then $|\delta(U)|$ is at most $2k$.
Suppose we want to bound the probability that
the sampled weight $w(\delta(U))$ exceeds $4k$.
For this to happen, at most $2\rho$ copies of $st$
can be included in the sparsifier.
By a Chernoff bound, this many edges are included 
with probability at most $e^{-2\rho}=e^{O(-\log^2n)}$.
However, there are $\Omega(2^n)$ such $U$'s,
which is too many for such a union bound to work.

The reason this union bound fails
is that the event ``more than $2\rho$ copies of $st$ are included''
is overcounted $\Theta(2^n)$ times, once for each $U$.
However since all $\delta(U)$'s share the same $k$ copies of $st$,
we actually only need to analyze this event once.

Bencz\'ur and Karger \cite{BK} accomplished this by decomposing the graph.
Assume that the edges $E=\{e_1,\ldots,e_m\}$ are sorted by increasing edge strengths.
Each $G_j$ contains all edges $e_i$ with $i\ge j$.
Then $G$ can be viewed as the sum of $G_j$'s,
with each $G_j$ scaled by $k'_{e_j}-k'_{e_{j-1}}$.
An important property of this decomposition is that
if $e_i$ is in $G_j$ 
then the strength of $e_i$ in $G_j$ is the same as
its in the original graph $G$. 
This is because the $k'_{e_i}$-strong component $H$ in $G$ that contains $e_i$
must also be present in $G_j$, as all edges in $H$ have strengths at least $k'_{e_i}$.

Therefore, even though edges in $G_j$
have small sampling probabilities (at most $1/k'_{e_j}$), 
the expected number of sampled edges in every cut is 
at least $\Omega(\log n/\epsilon^2)$,
since the min cut of $G_j$ is large (at least $k'_{e_j}$).
Thus Karger's graph skeleton analysis is applicable to sampling in $G_j$.
Roughly speaking, in order to use the Chernoff bound 
to obtain a constant factor approximation
in the number of sampled edges in a cut
with a failure probability of $n^{-\Omega(1)}$,
the expected number of sampled edges in the cut needs to be $\Omega(\log n)$.

To prove \Theorem{mainthm}, we could attempt to use the same decomposition to 
analyze our sampling scheme where edge connectivity is used instead of strength.
The problem is that in general
edge connectivity is not preserved under such decomposition.
To see this, consider the example in \Figure{example1}.
Observe that the subgraph induced by those edges 
with connectivities at least $n-1$ 
consists of only one edge $st$, so this subgraph has min cut value $1$.
The expected number of copies of $st$ in the sparsifier is $O(\log^2 n/n)$,
so we cannot expect to say that sampling preserves every cut of this subgraph
to within $1\pm \epsilon$.

\section{Proofs for \Section{badevents}}
\AppendixName{boundproofs}

In this section we prove \Claim{large_lambda}, \Claim{small_lambda} and \Claim{errorub}.
We require the following three versions of the Chernoff bound.
For the case $U=1$, these can be found in the survey of McDiarmid~\cite{McDiarmid};
the case of larger $U$ reduces to that case by scaling.

\begin{theorem}
\TheoremName{chernoff_3}
Let $X_1,\ldots,X_k$ be independent random variables with values in $[0,1]$.
Let $\alpha_1,\ldots,\alpha_k$ be scalars in $[0,U]$.
Let $X$ be a weighted sum of Bernoulli trials defined by
$X=\sum_{i=1}^{k} \alpha_i X_i$, and let $\mu = \expect{X}$.
Then for any $\delta > 0$, we have
$$
    \prob{ |X-\mu|\ge \delta\mu} ~\le~ 2\exp\Bigg(- \frac{ \delta^2 \mu}{2(1+\delta/3)U} \Bigg).
$$
\end{theorem}

\begin{corollary}
\CorollaryName{chernoff_1}
Let $X$ and $\mu$ be as in \Theorem{chernoff_3}.
Then for any $0<\delta<1$, we have
$$
    \Pr[|X-\mu|\ge \delta\mu]\le 2e^{-\delta^2 \mu/3U}.
$$
\end{corollary}

\begin{theorem}
\TheoremName{chernoff_4}
Let $X$ and $\mu$ be as in \Theorem{chernoff_3}.
Then for any $\delta > 0$, we have
$$
    \Pr[X-\mu\ge \delta\mu] ~\le~ \exp\Bigg(-\frac{g(\delta)\mu}{U} \Bigg),
$$
where $g(\delta) = (1+\delta) \ln( 1+\delta ) - \delta$ is the function defined in \Equation{gdef}.
\end{theorem}

\repeatclaim{\Claim{large_lambda}}{\clmlargelambda}

\begin{proof}
Let $U = \frac{2^{i+1}}{\rho}$ and $\mu=\expect{X_F}=|F|$.
By the definition of our sampling process,
$X_F$ is a weighted sum of Bernoulli trials where each weight is less than $U$.
Thus
\begin{align*}
\prob{\cA_F}
    &~=~ \prob{|X_F-|F|| > \epsilon |F|} \\
    &~\le~ 2\exp\Bigg(-\frac{\epsilon^2|F|}{3 U}\Bigg)
        \qquad\text{(by \Corollary{chernoff_1})}\\
    &~\leq~ 2\exp\Bigg(- \frac{\epsilon^2 q(F) \rho}{3 \ln(n) 2^{i+1}} \Bigg)
        \qquad\text{(by our assumption on $q(F)$)}\\
    &~\leq~ 2\exp\Bigg(- \frac{\epsilon^2 \alpha(F) d \ln n}{6} \Bigg) \\
    &~\le~ 2n^{-d\alpha(F)\epsilon^2/6}.
\end{align*}
This concludes the proof.
\end{proof}

\repeatclaim{\Claim{small_lambda}}{\clmsmalllambda}

\begin{proof}
Let $U = \frac{2^{i+1}}{\rho}$, $\mu=\expect{X_F}=|F|$
and $\delta=\frac{\epsilon q(F)}{\ln(n) |F|}$.
Then $X_F$ is a weighted sum of Bernoulli trials, and $U$ is an upper bound on the weights.
Note that $\delta \ge \epsilon$ and $1 \geq \epsilon$.
Thus
$$
    \frac{\delta^2}{1+\delta/3}
        ~\ge~ \frac{\delta\epsilon}{1+\epsilon/3}
        ~\ge~ \frac{\delta\epsilon}{2}.
$$
Also, 
\begin{equation}
 \EquationName{small_lambda}
    \frac{\delta\mu}{U} 
         ~=~ \frac{\epsilon q(F) \rho}{\ln(n) 2^{i+1}}
         ~=~ \frac{\epsilon \alpha(F) d \ln n}{2}.
\end{equation}
Thus
\begin{eqnarray*}
\prob{\cB_F} &=& \prob{ |X_F-|F|| \ge \frac{\epsilon q(F)}{\ln n} } \\
&=& \prob{|X_F-|F|| \ge \delta |F|} \\
&\le& 2\exp\Bigg(-\frac{\delta^2\mu}{2(1+\delta/3)U}\Bigg)
    \qquad\text{(by \Theorem{chernoff_3})}\\
&\le& 2\exp\Bigg(-\frac{\delta\epsilon\mu}{4 U}\Bigg) \\
&\le& 2\exp\Bigg(-\frac{\epsilon^2\alpha(F) d \ln n}{8}\Bigg)
    \qquad\text{(by \Equation{small_lambda})}\\
&=& 2n^{-d\alpha(F) \epsilon^2/8}.
\end{eqnarray*}
This concludes the proof.
\end{proof}

\repeatclaim{\Claim{errorub}}{\clmerrorub}

\begin{proof}
Let $U = \frac{ 2^{i+1} }{ \rho }$ and $\mu = \expect{ X_F } = |F|$.
Then
\begin{align*}
\prob{ \cC_F }
 &~=~ \prob{ X_F - |F| > g^{-1}\Big( \frac{ \epsilon^2 q(F) }{ |F| \ln n } \Big) \cdot |F| } \\
 &~\le~ \exp\Bigg(-\Big(\frac{ \epsilon^2 q(F) }{ |F| \ln n }\Big) \frac{|F|}{U} \Bigg)
    \qquad\text{(by Theorem \ref{thm:chernoff_4})} \\
 &~\le~ \exp\Bigg(-\frac{ \epsilon^2 q(F) \rho }{ \ln(n) 2^{i+1} } \Bigg) \\
 &~\le~ \exp\Bigg(-\frac{ \epsilon^2 \alpha(F) d \ln n }{ 2 } \Bigg).
\end{align*}
This completes the proof.
\end{proof}

\repeatclaim{\Claim{union_bound}}{\clmunionbound}

\begin{proof}
Fix an $i$ and let $F^1, F^2, \ldots$ be all the cut-induced subsets of $E_i$,
ordered such that $q(F^1) \le q(F^2) \le \ldots$.
Let
$$
p_j ~=~
\begin{cases}
\prob{ \cA_F \union \cC_F } &\quad\text{(if $q(F) \leq |F| \ln n$)} \\
\prob{ \cB_F \union \cC_F } &\quad\text{(if $q(F) > |F| \ln n$)}.
\end{cases}
$$
By Claims \ref{clm:large_lambda}, \ref{clm:small_lambda} and \ref{clm:errorub},
there exists a value $d=O(1/\epsilon^2)$ such that 
\begin{equation}
\EquationName{p_j}
p_j ~\le~ 4n^{-(r+2)\alpha(F^j)}.
\end{equation}
We consider the first $n^2$ cut-induced sets.
Note that for all $j \ge 1$, $p_j \le 4n^{-(r+2)}$.
Therefore, a union bound shows that
the probability that any bad event happens for some $F^j$ with $1 \leq j \leq n^2$
is at most $n^2 \cdot p_1 \le 4n^{-r}$.

Now we consider the remaining cut-induced sets $F^j \subseteq E_i$ for $j > n^2$.
\Corollary{count_canonical_cut} states that, for any $\alpha \geq 1$,
$$
\card{ \setst{ \text{cut-induced set } F \subseteq E_i }{ \alpha(F) \leq \alpha } } ~<~ n^{2\alpha}.
$$
Letting $\alpha>1$ be such that $j = n^{2 \alpha}$,
we see that $\alpha(F^j) > \alpha = \frac{\ln j}{2 \ln n}$.
Thus, from \Equation{p_j} we have
$$
p_j
    ~\le~ 4n^{-(r+2)\alpha(F^j)}
    ~<~ 4n^{-(r+2) \ln(j)/2 \ln n }
    ~\le~ 4j^{-(r+2)/2}
$$
Thus 
\begin{align*}
\sum_{j>n^2} p_j
    &~\le~ \sum_{j>n^2} 4j^{-(r+2)/2} \\
    &~\le~ \int^\infty_{n^2} 4j^{-(r+2)/2}\, dj \\
    &~=~ \frac{-8 j^{-r/2}}{r} \Bigg{|}_{n^2}^\infty \\
    &~=~ O(n^{-r}).
\end{align*}
This completes the proof.
\end{proof}

\section{Proofs for \Section{allcuts}}
\AppendixName{calcproof}

\repeatclaim{\Claim{Fsmall}}{\clmfsmall}

\begin{proof}
Since $F \subseteq E_i$, every $e \in F$ satisfies $k_e \leq 2^{i+1}$.
Since $u_e \leq k_e$ for every $e$, we obtain $u_e \leq 2^{i+1}$.
Thus $|F| = \sum_{e \in F} u_e \leq \binom{n}{2} 2^{i+1}$.
This proves the claim.
\end{proof}

\repeatclaim{\Claim{concentrate_in_large_Ci}}{\clmconc}

\begin{proof}
By Claim~\ref{clm:Fsmall},
$$
\sum_{i\le \lg |C|-2\lg n- d} \card{C \intersect E_i}
~<~ \sum_{i\le \lg |C|-2\lg n - d} n^2 2^i \\
~\le~ n^2 2^{\lg |C|-2\lg n -d +1} \\
~=~ 2^{1-d} \, |C|.
$$
This completes the proof.
\end{proof}

\repeatclaim{\Claim{ginv}}{\clmginv}

The purpose of this claim is to give a simple, asymptotically tight upper bound on $g^{-1}$.
We thank ``mathphysicist'' from the web site MathOverflow
for pointing out that a precise expression for $g^{-1}$ can be given using the Lambert $W$ function.
Specifically, one can show that
$$
g^{-1}(x) ~=~ 
\exp\Big( W\big(\smallfrac{x-1}{e}\big)+1\Big)-1.
$$
Unfortunately this exact expression is not terribly useful,
since we do not know of any simple, asymptotically tight bounds on $W$.

\vspace{6pt}
\begin{proofof}{\Claim{ginv}}
For all $x \geq 0$, we have
\begin{gather}
    \nonumber
    \sqrt{x} ~\geq~ \ln( 1 + \sqrt{x} )
    \\\nonumber\implies\qquad
    h(x) ~=~ \frac{2x}{\log( 1 + \sqrt{x} )} ~\geq~ \sqrt{x}
    \\\nonumber\implies\qquad
    \ln(1 + h(x)) ~\geq~ \ln(1 + \sqrt{x})
    \\\implies\qquad
    \EquationName{ginv1}
    h(x) \cdot \ln( 1 + h(x) )
        ~=~ \frac{ 2x \cdot \ln( 1 + h(x) ) }{ \ln( 1 + \sqrt{x} ) }
        ~\geq~ 2x.
\end{gather}
Next, for $y \geq 0$,
\begin{gather*}
    \frac{1}{1+y} ~\geq~ \frac{1}{1+y} - \frac{y}{(1+y)^2}.
    \intertext{Integrating, we obtain}
    \ln(1+y) ~\geq~ \frac{y}{1+y}
    \\\implies\qquad
    \ln(1+y) ~\geq~ \frac{1}{2} \Big( \ln(1+y) + \frac{y}{1+y}\Big).
    \intertext{Integrating again, we obtain}
    (1+y) \ln(1+y) - y ~\geq~ \frac{y \ln(1+y)}{2}.
\end{gather*}
Substituting $y=h(x)$ into this and applying \Equation{ginv1} yields
$$
g(h(x))
    ~=~ (1+h(x)) \ln(1+h(x)) - h(x)
    ~\geq~ \frac{h(x) \ln(1+h(x))}{2}
    ~\geq~ x.
$$
As observed above, $g^{-1}$ is strictly monotonically increasing.
Thus $g(h(x)) \geq x$ implies that $h(x) \geq g^{-1}(x)$.
\end{proofof}

\section{Probability of success in \Algorithm{contract}}
\AppendixName{contract_alg}

In this appendix we complete the proof of \Theorem{contract_with_split}
by proving the following claim.

\repeatclaim{\Claim{output}}{\clmoutput}

\begin{proof}
Define $R := \ceil{2 \alpha}$.
In the last iteration of the algorithm, the number of remaining vertices is at least $R+1$.
The probability that the invariants hold at the end of the algorithm
is at least the product of the probabilities that the 
invariants are not violated at any step of the algorithm.
By Claims \ref{clm:splitoff} and \ref{clm:contract}, this probability is at least
$$
    \prod_{r=n}^{R+1} 1 - \frac{2 \alpha}{r} 
    ~~=~~ \frac{n-2\alpha}{n} \cdot \frac{n-1 -2\alpha}{n-1} \cdots \frac{R+1 -2\alpha}{R+1} 
    ~~=~~ \frac{(n-2\alpha)!}{(R-2\alpha)!}\frac{R!}{n!},
$$
where the factorial function is extended to arbitrary real numbers via the Gamma function.

Since there are at most $R$ remaining vertices at the end of the algorithm
there are less than $2^{R-1}$ remaining non-trivial cuts,
at least one of which induces $F$, by (I1).
Therefore, the probability that the last step of the algorithm
selects a set $S$ that induces $F$ is at least
$$
    2^{1-R} \frac{(n-2\alpha)!}{(R-2\alpha)!}\frac{R!}{n!}
    ~>~ \frac{1}{n^{2\alpha}} \frac{1}{(R-2\alpha)!}
    ~\ge~ n^{-2\alpha},
$$
where we have used the inequalities $2^{R-1}\le R!$,
$n! / (n-2\alpha)! < n^{2\alpha}$,
and $x! \leq 1$ for $x \in [0,1]$.
\end{proof}

\section{Random contraction algorithm by random walks}
\AppendixName{random_walk_alg}

In this appendix we present \Algorithm{contractRW},
which is a variant of the contraction algorithm that contracts
random walks instead of random edges.
We use the similar terminology and notation to \Section{splitting_off_alg},
e.g., black vertices.
The following analysis of the algorithm immediately implies \Theorem{count_cut_2}.

\begin{theorem}
\TheoremName{contract_with_RW}
For any cut-induced set $F \subseteq B$ with $q(F) \leq \alpha K$,
\Algorithm{contract} outputs $F$ with probability at least $n^{-2\alpha}$.
\end{theorem}

\begin{algorithm}
\begin{alg}
    \item	\textbf{procedure} ContractRW($G$,\, $B$,\, $\alpha$)
    \item	\textbf{input:} A graph $G=(V,E)$, a set $B \subseteq E$,
            and an approximation factor $\alpha$
    \item	\textbf{output:} A cut-induced subset of $B$
	\item	While there are more than $\lceil 2 \alpha \rceil$ black vertices remaining
	\begin{alg}
        \item	Randomly pick a black vertex $u_1$ with probability proportional to its degree
                %$\frac{d(u_1)}{\sum_{w\in W}d(w)}$
        \item	Perform a random walk starting from $u_1$ and
                stopping when it hits a black vertex $u_2$ (possibly $u_1=u_2$)
        \item   If $u_1 \neq u_2$
        \begin{alg}
            \item	Contract all edges traversed by the random walk and remove any self loops
        \end{alg}
    \end{alg}
	\item	Pick a non-empty proper subset $S$ of $V$ uniformly at random
	        and output the black edges with exactly one endpoint in $S$
\end{alg}
\caption{An algorithm for finding a small cut-induced set by contracting random walks.}
\AlgorithmName{contractRW}
\end{algorithm}

The approach for proving \Theorem{contract_with_RW}
is again similar to Karger's analysis of the contraction algorithm
--- for any cut $C$, we can bound the probability that an edge in $C$ is contracted. 
Formally, our analysis is:

\begin{lemma}
\LemmaName{rw}
Consider an iteration of the while loop that begins with $r$ remaining black vertices.
Suppose that no edge in $C$ has been contracted so far.
Suppose that the random walk in this iteration has $u_1 \neq u_2$.
Then
$$
\Pr[ \text{\normalfont some edge in $C$ is contracted in this iteration} ]
    ~\leq~
    \frac{2 \card{C}}{K r}.
$$
\end{lemma}

To prove \Theorem{contract_with_RW}, one applies \Lemma{rw}
where $C$ is a cut which induces $F$ and satisfies $\card{C} \leq q(F) \leq \alpha K$.
The remainder of the proof follows by the same argument as \Claim{output}.

The key method in proving \Lemma{rw} is to understand the probability
that a random walk hits a certain set of vertices before hitting some other set of vertices.
To that end, let us introduce some notation.
For any two sets of vertices $X$ and $Y$, 
let $\cond_{X,Y}$ denote the effective conductance between $X$ and $Y$.
Equivalently, identify $X$ into a single vertex $x$, identify $Y$ into a single vertex $y$,
and let $\cond_{X,Y}$ be the effective conductance between $x$ and $y$.

Next, suppose that $s \in V$ and that $T$ and $U$ are disjoint subsets of $V$.
We use $s \rightarrow T < U$ to denote the event that 
a random walk starting at $s$ hits $T$ before it hits $U$.
If $T = \{t\}$, we use the shorthand $s \rightarrow t < U$,
and similarly if $U = \{u\}$.
In the case that $s \in T \cup U$, the term ``hits'' means
``hits after performing at least one step of the random walk''.

\begin{remark}
The event $s \rightarrow T < U$ can also be understood in another way.
Let $G'$ be the graph obtained by identifying all nodes in $T$ into a single node $t$,
and identifying all nodes in $U$ into a single node $u$.
Then $\Pr[s \rightarrow T < U]$ equals the probability that a random walk in $G'$
starting at $s$ hits $t$ before $u$.
\end{remark}

The main tool in the proof of \Lemma{rw} is the following reciprocity law.
It will allow us to consider random walks originating at the cut $C$ rather
than random walks that cross $C$.

\begin{lemma}
\LemmaName{switch}
Let $s, t \in V$ and $U \subseteq V$.
Assume that $s \neq t$ and $\{s,t\} \cap U = \emptyset$.
Then
\begin{equation}
\EquationName{recip}
        \cond_{s,\{t\} \cup U} \cdot \Pr[s\rightarrow t < U]
    ~=~ \cond_{t,\{s\} \cup U} \cdot \Pr[t \rightarrow s < U].
\end{equation}
\end{lemma}

To prove this lemma, we need to understand the relationship between the following events:
\begin{align*}
\cE   &~=~ s\rightarrow T < \{s\} \cup U \\
\cE'  &~=~ s\rightarrow T \cup U < s \\
\cE^* &~=~  s\rightarrow T < U.
\end{align*}
In English, $\cE$ is the event
that the random walk hits $T$ before hitting $U$ or returning to $s$,
$\cE'$ is the event that the random walk hits $T$ or $U$ before returning to $s$,
and $\cE^*$ is the same as $\cE$ except that
the random walk is permitted to return to $s$ before hitting $T$ or $U$.

\begin{claim}
\ClaimName{split_xi}
$\cE = \cE^* \wedge \cE'$.
\end{claim}

\begin{claim}
\ClaimName{indep_xi}
$\cE^*$ and $\cE'$ are independent.
\end{claim}
\begin{proof}
The claim essentially follows from the ``craps principle''.
In more detail, consider any random walk starting $s$ and ending at $T \cup U$.
It can be viewed as a sequence $w_1, \ldots, w_k$ of random walks where
\begin{itemize}
\item for each $i<k$, $w_i$ starts and ends at $s$ but otherwise does not traverse $s$,
\item $w_k$ starts at $s$ and ends at $T \cup U$ but otherwise does not traverse $s$.
\end{itemize}
So $\Pr[\cE^*]$ is the probability that $w_k$ ends at $T$,
and by the Markov property, this equals
$$ \Pr[ s \rightarrow T < U \mid s \rightarrow T \cup U < s ]. $$
Thus we have argued that $ \Pr[\cE^*] = \Pr[\cE^* \mid \cE']$, as required.
\end{proof}

\begin{claim}
\ClaimName{split_xi_2}
$\Pr[\cE'] = \cond_{s,T \cup U}/d(s)$.
\end{claim}
\begin{proof}
Doyle and Snell~\cite[\S 1.3.4]{DoyleSnell}.
\end{proof}

\vspace{12pt}

\begin{proofof}{\Lemma{switch}}
It is known \cite[Theorem IX.22]{bollobas} that
$$
    d(s) \cdot \Pr[ s \rightarrow t < \{s\} \cup U ]
    ~=~ d(t) \cdot \Pr[ t \rightarrow s < \{t\} \cup U ].
$$
By \Claim{split_xi} and \Claim{indep_xi}, this is equivalent to
$$
    d(s) \cdot \Pr[ s \rightarrow \{t\} \cup U < s ] \cdot \Pr[ s \rightarrow t < U ]
    ~=~ 
    d(t) \cdot \Pr[ t \rightarrow \{s\} \cup U < t ] \cdot \Pr[ t \rightarrow s < U ].
$$
By \Claim{split_xi_2}, this is equivalent to \eqref{eq:recip}.
\end{proofof}

\vspace{12pt}

\begin{proofof}{\Lemma{rw}}
It is more convenient to consider hitting a vertex than a cut,
so we subdivide every edge in $C$ and merge the subdividing vertices into a new vertex $z$.
We consider the random walk in the modified graph induced by the random walk in the
original graph.
The former walk hits $z$ iff the latter walk intersects $C$. 

For the remainder of this proof, $W$ denotes the set of currently remaining black vertices.

\begin{claim}
\ClaimName{uw}
For any $w \in W$,
$$
    \Pr[u_1=w \mid u_1 \neq u_2]
        ~=~ \frac{ \cond_{w,W\setminus \set{w}} }{ \sum_{u \in W} \cond_{u,W\setminus \set{u}} }.
$$
\end{claim}
\begin{subproof}
Let $D = \sum_{v \in W} d(v)$.
Then
\begin{align*}
\Pr[u_1=w \mid u_1 \neq u_2] 
&~=~ \frac{ \Pr[u_1=w ~\wedge~ u_1 \neq u_2] }
    { \sum_{u\in W} \Pr[u_1=u ~\wedge~ u_1 \neq u_2] } \\
&~=~ \frac{ \Pr[u_1=w ]
    \cdot \Pr[ u_1 \neq u_2 \mid u_1=w ] }
    { \sum_{u\in W} \Pr[u_1=u] \cdot \Pr[ u_1 \neq u_2 \mid u_1=u ] } \\
&~=~ \frac{ \big( d(w)/D \big)
    \cdot \Pr[ w \rightarrow W \setminus \{w\} < w ] }
    { \sum_{u\in W} \big( d(u)/D \big) \cdot \Pr[ u \rightarrow W \setminus \{u\} < u ] } \\
&~=~ \frac{ d(w)
    \cdot \big( \cond_{w, W \setminus \{w\}} / d(w) \big) }
    { \sum_{u\in W} d(u) \cdot \big( \cond_{u, W \setminus \{u\}} / d(u) \big) }
    \qquad\qquad\text{(by \Claim{split_xi_2})}
    \\
&~=~ \frac{ \cond_{w,W\setminus \{w\}} }
    { \sum_{u \in W} \cond_{u,W\setminus \{u\}} }.
\end{align*}
This proves the claim.
\end{subproof}

\begin{claim}
\ClaimName{given}
Let $\cA$ and $\cB$ be disjoint events. Let $\cC$ be another event.
$$
\Pr[ \cC \mid \cA ] \geq \Pr[ \cC \mid \cA \vee \cB ]
\quad\implies\quad
\Pr[ \cC \mid \cB ] \leq \Pr[ \cC \mid \cA \vee \cB ].
$$
\end{claim}

\begin{claim}
\ClaimName{mess}
For any $w \in W$,
$$
    \Pr[ w \rightarrow z < W \setminus \set{w} ~\mid~ w \rightarrow W \setminus \set{w} < w ]
    ~~\leq~~
    \Pr[ w \rightarrow z < W \setminus \set{w} ].
$$
\end{claim}
\begin{subproof}
Define
\begin{align*}
\cB &~=~ w \rightarrow W \setminus \set{w} < w \\
\cA &~=~ (w \rightarrow z < w) \setminus \cB \\
\cC &~=~ w \rightarrow z < W \setminus \set{w}.
\end{align*}
Since $\Pr[ \cC \mid \cA ] = 1$, the hypotheses of Claim~\ref{clm:given} are satisfied,
and therefore 
$$
\Pr[ w \rightarrow z < W \setminus \set{w} \mid w \rightarrow W \setminus \set{w} < w ]
~\leq~ 
\Pr[ w \rightarrow z < W \setminus \set{w} \mid w \rightarrow (W \setminus \set{w}) \union \set{z} < w ].
$$
By \Claim{indep_xi}, the latter quantity equals
$ \Pr[ w \rightarrow z < W \setminus \set{w} ]$.
\end{subproof}

\vspace{12pt}

Now, we analyze the probability that the random walk hits the cut.
We condition on the event $u_1 \neq u_2$, since that is the only case when
the algorithm contracts edges.
\begin{eqnarray*}
&&\Pr[\text{random walk hits } z \mid u_1 \neq u_2] \\
&=& 
\sum_{w\in W} \Pr[u_1=w \mid u_1 \neq u_2]
    \cdot \Pr[ \text{ random walk hits } z ~\mid~ u_1=w ~\wedge~ u_1 \neq u_2] \\
&=& 
\sum_{w\in W} \Pr[u_1=w \mid u_1 \neq u_2]
    \cdot \Pr[\: w \rightarrow z < W \setminus \{w\} ~\mid~ w \rightarrow W \setminus \set{w} < w \:] \\
&=& 
\sum_{w\in W} \frac{ \cond_{w,W \setminus \set{w}} }
                   { \sum_{u \in W} \cond_{u,W \setminus \set{u}} }
    \cdot \Pr[\: w \rightarrow z < W \setminus \set{w} ~\mid~ w \rightarrow W \setminus \set{w} < w \:]
    \qquad\text{(by \Claim{uw})} \\
&\leq& 
\sum_{w\in W} \frac{ \cond_{w,W\setminus \set{w}} }
                   { \sum_{u \in W} \cond_{u,W \setminus \set{u}} }
    \cdot \Pr[ w \rightarrow z < W \setminus \set{w} ]
    \qquad\text{(by \Claim{mess})} \\
&\leq& 
\sum_{w\in W} \frac{ \cond_{w,\set{z} \union (W \setminus \set{w})} }
                   { \sum_{u \in W} \cond_{u,W\setminus \set{u}} }
    \cdot \Pr[ w \rightarrow z < W \setminus \set{w} ] \\
&=& \sum_{w\in W} 
\frac{ \cond_{z,W} }
{ \sum_{u\in W} \cond_{u,W\setminus \set{u}} }
\Pr[ z \rightarrow w < W \setminus \set{w} ]
    \qquad\text{(by Lemma~\ref{lem:switch})}\\
&=&  
\frac{ \cond_{z,W} }
{ \sum_{u\in W} \cond_{u,W\setminus \set{u}} }
\underbrace{\sum_{w\in W} \Pr[ z \rightarrow w < W \setminus \set{w} ]}_{
    \text{($=1$, because the walk is ergodic)} } \\
&=&  
\frac{ \cond_{z,W} }
{ \sum_{u\in W} \cond_{u,W\setminus \set{u}} } \\
&\le& \frac{2|C|}{K |W|}.
\end{eqnarray*}
The last inequality is because the node $z$ has degree $2 |C|$,
and every node $u \in W$ is connected to some node $v \in W \setminus \set{u}$
by a black edge, and $\cond_{u,W \setminus \set{u}} \geq \cond_{u,v} \geq K$.
\end{proofof}

\section{Algorithms for constructing sparsifiers}
\AppendixName{SketchAlgorithms}

In this section,
we sketch the algorithms as stated in Theorems \ref{thm:alg1}, \ref{thm:alg2} and \ref{thm:alg3}.
They are simple modifications of the algorithms of Bencz\'ur and Karger \cite{BK}.
The main difference is that Bencz\'ur and Karger's algorithms compute approximate edge strengths
whereas our modifications compute approximate edge connectivities.
The sparsifiers we obtain have slightly larger size but the algorithms are simpler and more efficient
because approximating the edge connectivities is quite a bit easier
than approximate the edge strengths.
Our algorithms can be easily implemented, and furthermore they can be used as a preprocessing step
for computing smaller sparsifiers.
The proofs of correctness of these algorithms are almost exactly the
same as the proofs in \cite{BK}.

\subsection{Finding $O(n\log^3 n)$-size sparsifiers for graphs with polynomial weights}

We now present an algorithm that computes a sparsifier of size $O(n\log^3 n)$.
It runs in $O(m)$ time for unweighted graphs and 
$O(m+n\log^5n)$ time for graphs with polynomial weights.
Recall that, in \Theorem{mainthm},
it is sufficient to find $\kappa_e$ 
that is a lower bound of the edge connectivity $k_e$.

The main tool that we use is the \emph{$k$-certificate}
introduced by Nagamochi and Ibaraki \cite{NIForest1}.
Given a multigraph $G=(V,E)$, they partition $E$ into a set of forests
$F_1,F_2,\ldots$ in the follwing way.
Let $F_1$ be a maximal forest of $G$
and for $i>1$, let $F_i$ be a maximal forest of the subgraph $G-\cup_{j<i} \, F_j$.
Each $F_i$ is called a \emph{NI-forest}.
Nagamochi and Ibaraki showed that for any integer $k\ge 0$,
$H_k=\cup_{j\le k} F_j$,
the union of the first $k$ NI-forests, preserves all cuts of $G$
that have size at most $k$.
Thus $H_k$ contains all $(k+1)$-light edges of $G$
(an edge $e$ is $k$-heavy if $k_e\ge k$ and $k$-light otherwise).
$H_k$ is called a $k$-certificate.
Clearly, it contains at most $k(n-1)$ edges.

Nagamochi and Ibaraki \cite{NIForest1} presented an algorithm for labeling 
every edge $e$ with a label $r_e$,
such that if an edge $e$ has multiplicity $u_e>1$, the $u_e$ copies of $e$
are contained in the $u_e$ NI-forests $F_{r_e-u_e+1},\ldots,F_{r_e}$.
For simple graphs, it runs in $O(m)$ time.
For multigraphs, a slightly modified version \cite{NIForest2} of this algorithm runs in
$O(m+n\log n)$ time.

Let $e=st$ be an edge.
Note that if $e$ appears in $F_i$, 
then $s$ and $t$ must be connected in $F_j$ for every $j<i$,
for otherwise $e$ can be added to $F_j$,
which contradicts the maximality of $F_j$.
Therefore, $r_e$ is a lower bound of $k_e$
and we can set $\kappa_e=r_e$.
Suppose we sample with probability $p_e=1/\kappa_e$ as described in \Theorem{mainthm}.
Then with high probability, this will produce a sparsifier that
preserves every cut to within a $1\pm \epsilon$ factor.
Since we assume the weights are polynomially bounded,
the expected number of edges per round is 
$$
    \sum_{e\in E}u_e/r_e
    \le \sum_{j=1}^{n^2 \poly(n)}\sum_{e\in F_j}1/j
    \le (n-1)\sum_{j=1}^{n^2 \poly(n)}1/j
    = O(n\log n).
$$
Therefore with high probability,
the sparsifier contains $O(n\log^3n/\epsilon^2)$ edges.

Using the Nagamochi-Ibaraki Certificate algorithm,
all $\kappa_e$ can be found in $O(m+n\log n)$ time.
For each edge $e$, we can decide in expected constant time 
whether at least one copy of $e$ is included in the sparsifier.
For an edge that has at least one copy in the sparsifier,
we can find the sum of weights of all copies of it 
by sampling from the distribution Binomial($u_e \rho$, $p_e$)
instead of sampling $u_e$ Bernoulli random variables per round.
This sampling is easy to do in
$O(u_e \rho p_e)=O(u_e\log^2 n/k_e\epsilon^2)=O(\log^2 n/\epsilon^2)$ time,
for each edge included in the sparsifier.
We suspect that this can be improved using the technique of Knuth and Yao
\cite{KnuthYao} \cite[Ch.~15]{Devroye},
but leave the details to future work.
Therefore, the total running time is $O(m+n\log^5n/\epsilon^4)$ for graphs with polynomial weights.

For unweighted  graphs, we can reduce this to $O(m)$ time.
Note that for unweighted graphs, $u_e=1$ for all $e\in E$.
If an edge $e$ has $\rho/k_e>1$,
instead of performing $\rho$ rounds of sampling,
we can include $e$ with probability $1$ and assign it a weight of $1$.
Thus we can sample the weight of an edge in expected constant time.
This change can only decrease the expected size of the sparsifier
and the sampled weight of every cut can only be more concentrated around its mean.

We remark that recent work of Hariharan and Panigrahi \cite{HP}
analyzes the same algorithm and shows that 
actually setting $\rho = O(\log(n)/\epsilon^2)$ is sufficient,
whereas we set $\rho = O(\log^2(n)/\epsilon^2)$.
Therefore the size of their sparsfier is only $O(n\log^2(n)/\epsilon^2)$.

\subsection{Finding $O(n\log^2 n)$-size sparsifiers for graphs with polynomial weights}

In this section, we describe another algorithm for computing $\kappa_e$
which has the advantage that the computed $\kappa_e$'s
satisfy $\sum_{e\in E}u_e\kappa_e=O(n)$. 
By the last statement of \Theorem{mainthm},
the size of the sparsifier would be $O(n\log^2 n)$.
The algorithm, given in \Algorithm{connest},
is a slight variation of the Estimation algorithm in \cite{BK}.

\begin{algorithm}
\begin{alg}
	\item	\textbf{procedure} ConnectivityEstimation($H$,$k$)
	\item	\textbf{input:} subgraph $H$ of $G$
	\item	$E'\leftarrow$ Partition($H$,$2k$)
	\item	for each $e\in E'$
        \begin{alg}
        \item	$\kappa_e\leftarrow k$
        \end{alg}
	\item	for each nontrivial connected component $H' \subset H-E'$
        \begin{alg}
    	\item	ConnectivityEstimation($H'$,$2k$)
        \end{alg}
\end{alg}
\caption{Algorithm for estimating edge connectivities.}
\AlgorithmName{connest}
\end{algorithm}

The algorithm is based on finding $k$-partitions.
A \emph{$k$-partition} of a graph $G=(V,E)$
is a set $E'\subseteq E$ of edges that includes all $(k+1)$-light edges
such that $|E'|\le 2 k(r-1)$ if $G-E'$ has $r$ components.
A $k$-partition is a ``sparser'' version of $k$-certificate
as a $k$-certificate can have $k(n-1)$ edges.

\begin{lemma}[Bencz\'ur and Karger \cite{BK}]
\label{alg:partition}
There is an algorithm Partition that outputs 
a $k$-partition in $O(m)$ time for unweighted graphs
and $O(m\log n)$ time for graphs with arbitrary weights.
\end{lemma}

We compute the $\kappa_e$'s by using the ConnectivityEstimation procedure below.
It is almost the same as the Estimation procedure in \cite{BK},
which is for finding lower bounds of edge strengths.
The only difference is that they call a WeakEdges procedure
to find the $k$-weak edges
instead of calling Partition to find $k$-light edges.

\begin{lemma}
After a call to ConnectivityEstimation($G$,$1$),
all the labels $\kappa_e$ satisfy $\kappa_e\le k_e$.
\end{lemma}
\begin{proof}
The proof is the same as the proof of Corollary 4.8 in \cite{BK}.
\end{proof}

\begin{lemma}
\label{lem:sumKappa}
The values $\kappa_e$ output by ConnectivityEstimation satisfy $\sum 1/\kappa_e = O(n)$.
\end{lemma}
\begin{proof}
The proof is the same as the proof of Lemma 4.9 in \cite{BK}.
\end{proof} 

\begin{lemma}
ConnectivityEstimation runs in $O(m\log^2n)$ time on a graph with polynomial weights.
\end{lemma}
\begin{proof}
For a graph with polynomial weights, 
the maximum connectivity is bounded by some fixed polynomial,
so there are at most $O(\log n)$ levels of recursion.
The total number of edges in all the input graphs to Partition at each level of recursion
is at most $m$.
By Lemma \ref{alg:partition}, each level of computation takes $O(m\log n)$ time. 
\end{proof}

Suppose we sample as described in \Theorem{mainthm}.
By \Theorem{mainthm} and Lemma \ref{lem:sumKappa},
the size of the sparsifier is $O(n\log^2n/\epsilon^2)$.
The time needed for finding $\kappa_e$'s is $O(m\log^2n)$
and the time needed for finding the weight of every edge is $O(n\log^4n/\epsilon^4)$.
Therefore, the total running time is $O(m\log^2n+n\log^4n/\epsilon^4)$.

\subsection{Finding $O(n\log n)$-size sparsifiers for graphs with polynomial weights}

In \cite{BK}, Bencz\'ur and Karger presented an algorithm that finds a sparsifier of size 
$O(n\log n/\epsilon^2)$ for graphs with arbitrary weights in $O(m\log^3 n)$ time. 
This can be combined with the algorithms in the last two sections
to prove Theorem \ref{thm:alg3}.

Suppose we are given a graph $G$.
First we apply Theorem \ref{thm:alg1} to find a sparsifier $G'$
which approximately preserves all cuts of $G$ to within a multiplicative error of $1\pm \epsilon/4$.
$G'$ has size $O(n\log^3n/\epsilon^2)$ and
this takes $O(m+n\log^5n/\epsilon^4)$ time. 
Then we apply Theorem \ref{thm:alg2} to find a sparsifier $G''$
which is a $1\pm\epsilon/4$-approximation of $G'$.
$G''$ has size $O(n\log^2n/\epsilon^2)$
and this takes $O(n\log^3n\cdot \log^2n/\epsilon^2+n\log^4 n/\epsilon^4)
=O(n\log^5n/\epsilon^4)$ time.
Finally, we use Bencz\'ur and Karger's algorithm to 
obtain a sparsifier $G'''$ that is a $1\pm \epsilon/4$-approximation of $G''$.
$G'''$ has size $O(n\log n/\epsilon^2)$ and
this takes $O(n\log^2n/\epsilon^2 \cdot \log^3 n)=O(n\log^5n/\epsilon^2)$ time. 

Note that $G'''$ approximately preserves all cuts of $G$ 
to within a multiplicative error of $1\pm \epsilon$.
The total running time is $O(m+n\log^5n/\epsilon^4)$.

\end{document}